\documentclass[%
 aps, reprint, pra,
 amsmath, amssymb,
floatfix,
]{revtex4-2}

\usepackage{xcolor}
\usepackage{graphicx}
\usepackage{amsthm}
\usepackage[ruled,vlined,linesnumbered]{algorithm2e}
\usepackage{hyperref}
\usepackage{tabularx}
\usepackage{braket}
\usepackage{mathtools}

\newtheorem{theorem}{Theorem}

\newtheorem{lemma}[theorem]{Lemma}
\newtheorem{proposition}[theorem]{Proposition}
\newtheorem{definition}{Definition}

\DeclareMathOperator{\tr}{tr}
\DeclareMathOperator{\rank}{rank}
\DeclareMathOperator{\poly}{poly}
\DeclareMathOperator{\argmin}{arg\,min}

\allowdisplaybreaks

\hypersetup{colorlinks,
citecolor=blue,
linkcolor=red,
urlcolor=purple,
breaklinks=true}

\graphicspath{{./figures}}

\begin{document}

\title{Post-variational quantum neural networks}
\author{Po-Wei Huang}
\email{huangpowei22@u.nus.edu}
\affiliation{Centre for Quantum Technologies, National University of Singapore, 3 Science Drive 2, Singapore 117543}
\affiliation{Department of Computer Science, National University of Singapore, 13 Computing Drive, Singapore 117417}

\author{Patrick Rebentrost}
\email{patrick@comp.nus.edu.sg}
\affiliation{Centre for Quantum Technologies, National University of Singapore, 3 Science Drive 2, Singapore 117543}
\affiliation{Department of Computer Science, National University of Singapore, 13 Computing Drive, Singapore 117417}

\date{\today}

\begin{abstract}
Hybrid quantum-classical computing in the noisy intermediate-scale quantum (NISQ) era with variational algorithms can exhibit barren plateau issues, causing difficult convergence of gradient-based optimization techniques. In this paper, we discuss ``\emph{post-variational strategies}'', which shift tunable parameters from the quantum computer to the classical computer, opting for ensemble strategies when optimizing quantum models. We discuss various strategies and design principles for constructing individual quantum circuits, where the resulting ensembles can be optimized with convex programming. Further, we discuss architectural designs of post-variational quantum neural networks and analyze the propagation of estimation errors throughout such neural networks. Finally, we show that empirically, post-variational quantum neural networks using our architectural designs can potentially provide better results than variational algorithms and performance comparable to that of two-layer neural networks.
\end{abstract}

\maketitle

\section{Introduction}

Variational quantum methods~\citep{cerezo2021variational} are proposed to solve optimization problems in chemistry~\citep{peruzzo2014variational, mcclean2016theory}, combinatorial optimization~\citep{farhi2014quantum, hadfield2019quantum} and machine learning~\citep{mitarai2018quantum, schuld2020circuit, farhi2018classification, benedetti2019parameterized, schuld2019quantum} with a potential quantum advantage in the NISQ regime~\citep{preskill2018quantum, bharti2022noisy2}. These methods often use hardware-efficient Ans\"atze that are problem-agnostic~\citep{kandala2017hardware}. However, many Ans\"atze face the well-studied barren plateau problem \citep{mcclean2018barren, cerezo2021cost, ortiz_marrero2021entanglement, holmes2022connecting, wang2021noise, ragone2023unified, fontana2023adjoint}. Several methods are suggested to alleviate this problem, including parameter initialization that form block identities~\citep{grant2019initialization}, layerwise training~\citep{skolik2021layerwise}, correlating parameters in the Ansatz~\citep{volkoff2021large}, pretraining with classical neural networks~\citep{verdon2019learning}, and the usage of specific architectures such as quantum convolutional neural networks~\citep{cong2019quantum, pesah2021absence}. 

Partially inspired by the barren plateau problem, \citet{huang2021near} considered the solving of linear systems with near-term quantum computers with the use of \emph{classical combinations of quantum states} (CQS). Here, based on the input matrix, a so-called Ansatz tree is constructed which allows for a solution with provable guarantees and also for the systematic application of heuristic methods. The concept of utilizing combinations of quantum states and systematically generate Ans\"atze serves as an viable alternative to variational solutions that can circumvent the barren plateau problem, and has been utilized in quantum eigensolvers~\citep{bharti2020quantum, bharti2021iterative}, quantum simulation~\citep{bharti2021quantum, haug2022generalized, lau2021noisy, lau2022nisq, lim2021fast},  semidefinite programming~\citep{bharti2022noisy1} and finding non-equilibrium steady states~\citep{lau2023convex}. It is interesting to consider such concepts for problem agnostic settings such as neural networks for quantum machine learning~\citep{biamonte2017quantum} .

In our work, we discuss strategies derived from the variational method as the theoretical basis for optimization but avoiding the usage of parameterized quantum states. In what can be called  ``\emph{post-variational strategies}'', we take the classical combination of multiple fixed quantum circuits and find the optimal combination through solving a classical convex optimization problem~\citep{boyd2004convex}. 

Our post-variational quantum neural networks are based on ensemble strategies as opposed to optimizing parameterized quantum circuits. We exchange expressibility~\citep{du2020expressive, schuld2021effect, abbas2021power, jager2023universal} of the circuit with trainability of the entire model. Further, our method enforces a guarantee of finding a global minimum over our constructed convex optimization landscape in addition to ensuring termination of the algorithm. 

For the technical contributions, we generalize the idea of using linear combinations of unitaries to using linear combinations of quantum observables/quantum circuits with measurements. Based on this idea, we then propose two heuristic methods, as well as their hybrid, to construct post-variational quantum neural networks. Further, we show that our methods are compatible with classical shadows~\citep{huang2020predicting}, which replaces measurements on multi-qubit observables with one-shot random measurements. As a proof of concept, we employ our post-variational quantum neural network on the classical machine learning task of image classification, and obtain training performance better than that of variational algorithms and comparable to that of two-layer feedforward classical neural networks.

\section{Preliminaries}
\subsection{Notations} Given a field $\mathbb{F}$ of either real or complex numbers, for vector $\boldsymbol v \in \mathbb{F}^N$, we denote $\|\boldsymbol v\|_p$ as its $\ell_p$ norm. We denote $[h]$ to be the set $\{1, 2, \cdots, h\}$.

For a matrix $A \in \mathcal M_{M\times N}(\mathbb{F})$, let $A_{ij}$ be the $(i, j)$-element of $A$. We denote the spectral norm to be $\|A\| = \max_i \sigma_i(A)$, where $\sigma_i(A)$ are singular values of $A$, and the max norm to be $\|A\|_{\max} = \max_{ij} |A_{ij}|$. Furthermore, we denote $\sigma_{\min}(A)$ to be the smallest non-zero singular value of $A$. We denote $A^\intercal$ to be the transpose, $A^\dagger$ to be the conjugate transpose, and $A^+$ to be the pseudoinverse. Note that $\|A^+\| = \frac{1}{\sigma_{\min}(A)}$.

In this paper, we also refer to a variety of different loss functions that are used in different machine learning tasks. Given $d$ data points with ground truth $\{y_i\}_{i=1}^d$ and predicted values $\{\hat{y}_i\}_{i=1}^d$, where $\forall i, y_i, \hat{y}_i \in \mathbb{R}$, the root-mean-square error (RMSE) loss is defined as follows: $\mathcal{L}_{\mathrm{RMSE}} = \frac{1}{\sqrt{d}}\|\boldsymbol{y} - \boldsymbol{\hat{y}}\|_2$, while the mean absolute error (MAE) loss is defined by $\mathcal{L}_{\mathrm{MAE}} = \frac{1}{d}\|\boldsymbol{y} - \boldsymbol{\hat{y}}\|_1$. 

In the case where the ground truth is binary, i.e. $\forall i, y_i \in \{0, 1\}$ and predictions $\hat{y}_i \in [0, 1]$, the binary cross-entropy (BCE) loss is defined as $\mathcal{L}_{\mathrm{BCE}} = \frac{1}{d}\sum_{i=1}^d - y_i \log(\hat{y}_i) - (1-y_i) \log (1-\hat{y}_i)$.

\subsection{Classical shadows}The classical shadows method~\citep{huang2020predicting} introduces a randomized protocol to estimate the value of $\tr(O_i\rho)$ over $M$ observables up to additive error $\epsilon$ with $\mathcal O(\log(M))$ measurements. The quantum state is first evolved by a random unitary $U$ selected from a tomographically complete set of unitaries, i.e. $\rho \to U\rho U^\dagger$. Measuring the transformed quantum state in the computational basis, a string of outcomes $b\in \{0,1\}^n$ can be produced. One can then construct and store $U^\dagger|b\rangle\langle b|U$ classically using the stabilizer formalism~\citep{gottesman1997stabilizer}, the expectation of which can be viewed as a quantum channel $\mathcal{M}$, i.e.,
$\mathbb{E}[U^\dagger|b\rangle\langle b|U] = \mathcal{M}(\rho).$ By inverting the quantum channel and repeating the above process multiple times, we can obtain multiple ``\emph{classical shadows}'' $\hat{\rho}$ of the density matrix $\rho$ such that $\hat{\rho} = \mathcal{M}^{-1}(U^\dagger|b\rangle\langle b|U)$. These classical shadows can be used to approximate the value of the quantum state $\tr(O_i\rho)$ against a series of observables $O_1, O_2, \cdots, O_M$ via the median-of-means estimator~\citep{nemirovsky1983problem}. The number of classical shadows required to estimate $M$ observables within an additive error of $\epsilon$ is $\mathcal O(\log M \max_i\|O_i\|^2_{\mathrm{shadow}}/\epsilon^2)$, where $\|O_i\|^2_{\mathrm{shadow}}$, or the shadow norm, is dependent on the unitary ensemble used. When using tensor products of single-qubit Clifford gates (which is equivalent to measurement on a Pauli basis), $\|O_i\|^2_{\mathrm{shadow}} \le 4^L \|O_i\|^2$ for observable $O_i$ that acts non-trivially on $L$ qubits. For the rest of this paper, we refer to this particular property of an observable to be its \emph{locality}, and denote the shadow norm as $\|\cdot\|_S$.

\section{Post-variational strategies}

\subsection{Problem setting} We are given a dataset, $\mathcal{D} = \{(\boldsymbol{x_i}, y_i)\}^d_{i=1},$ where $d$ is the number of data, the feature vectors are ${\boldsymbol{x_i}} \in \mathbb R^\ell$, with $\ell$ being the number of features, and the labels $y_i \in \mathbb R$. Our task is to learn an estimator $\mathcal E_\vartheta (\boldsymbol{x_i})$ via parameterized neural networks with parameters $\vartheta$ that make use of quantum circuits such that $\hat y_i := \mathcal E_\vartheta (\boldsymbol{x_i})$. We aim to minimize the difference between the estimator $\hat{y_i}$ and ground truth $y_i$ over all data with a given loss function $\mathcal{L}$.

\subsection{Variational methods} Variational quantum algorithms have been regarded as the analogue to neural networks in quantum systems, and are also referred to as quantum neural networks (QNNs) when applied to machine learning tasks. Referring to a class of circuit-centric variational quantum algorithms operating on pure states~\citep{schuld2020circuit, farhi2018classification}, such algorithms operate by first encoding data $\boldsymbol{x}$ into a $n$-qubit quantum state $\rho(\boldsymbol{x}) \in \mathcal M_{2^n\times 2^n}(\mathbb{C})$. The quantum state is then transformed by an Ansatz $U(\boldsymbol{\theta})$, with parameters $\boldsymbol{\theta} \in \mathbb R^k$, to form a trial quantum state 
\begin{equation}
\varrho(\boldsymbol{\theta}, \boldsymbol{x}) = U(\boldsymbol{\theta})\rho({\boldsymbol{x}})U^\dagger(\boldsymbol{\theta}).
\end{equation}
We can then construct the estimator with an observable $O$ such that 
\begin{equation}
\mathcal E_{\boldsymbol{\theta}}(\boldsymbol{x}) := \tr(O\varrho(\boldsymbol{\theta}, \boldsymbol{x})).
\end{equation}
The parameters $\boldsymbol{\theta}$ are optimized by evaluating gradients of the quantum circuit via parameter-shift rules~\citep{mitarai2018quantum, schuld2019evaluating} and calculating updates of the parameter via gradient-based optimization on classical computers.

Apart from the simple construction of variational algorithms stated above that uses a single layer of data encoding followed by a variational Ansatz, there are variational algorithms make use of alternating data encoding layers and parameterized Ans\"atze, also known as data re-uploading models~\citep{perezsalinas2020data}. Such models are shown to have an exact mapping to the simpler construction of variational algorithms as shown above~\citep{jerbi2023quantum}, albeit with an exponential increase of qubits. In our work, as we do not specify the data encoding layer, we discuss the simpler construction by~\citet{schuld2020circuit} as a general structure that encompasses the mapped versions of data re-uploading models as well.

\begin{figure*}
    \centering
    \includegraphics[width=0.9\linewidth]{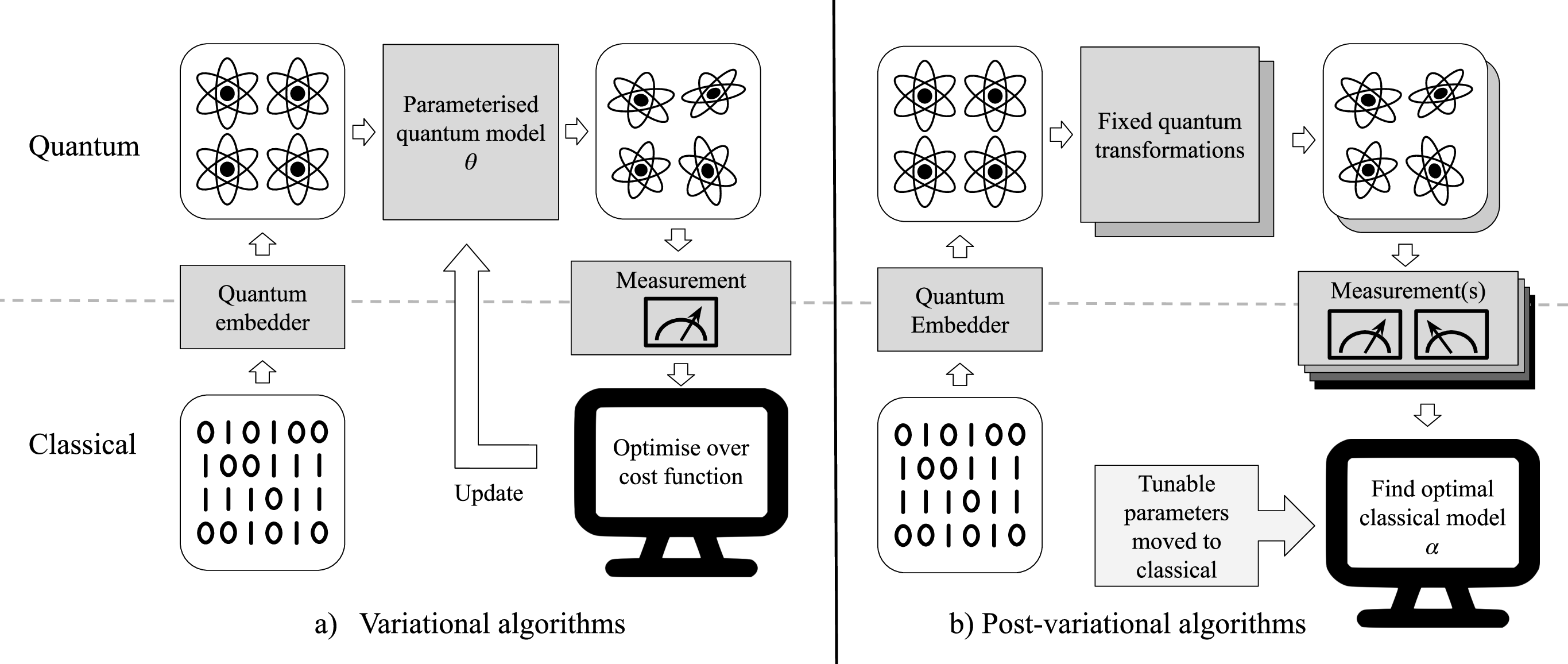}
    \caption{High-level sketch of post-variational strategies for near-term quantum computing. The variational method uses parameterized quantum circuits to transform embedded input data before suitable measurements, see panel a). Post-variational strategies, see panel b), use multiple fixed circuits, which may share similar circuit structure as the variational circuits, and multiple measurements of these circuits. The goal is to achieve approximately similar accuracy as the variational methods, while only performing optimization of classical parameters and retaining the power of quantum embeddings. }
    \label{fig:post_variational}
\end{figure*}

\subsection{Post-variational strategies} 

With the general difficulty and lack of training gurantees provided by variational algorithms, we propose what can be called ``post-variational'' strategies, where we replace parameterized quantum models with fixed quantum models and find the optimal combination of these models. In this setup, the quantum circuits are parameterized from the data only; all variational parameters are removed from the quantum circuits and trained classically. See Figure~\ref{fig:post_variational} for a high-level sketch of the idea behind post-variational strategies and Table~\ref{tab:compare} for a comparison of variational and post-variational strategies. The table is analogous to the works of \citet{huang2021near} and \citet{lau2023convex}, which are in the context of linear systems and non-equilibrium physics, respectively.

Note that the trainability of variational algorithms depends highly on the optimization landscape as dictated by the depth and width of the parameterized quantum circuit chosen, as well as the initial parameters selected for optimization. For the machine learning setting, our post-variational method exchanges the expressibility of the circuit with a guarantee that our algorithm can obtain the minima on the new convex optimization landscape we construct. We stress that we do not find the global minima on the variational optimization landscape or solve the barren plateau and equivalent problems such as exponential concentration~\citep{thanasilp2022exponential}, but instead provide an alternative approach to quantum neural networks. Later on in the paper, we will introduce heuristic design principles for selecting the fixed quantum circuits that aid in finding better solutions. 

\begin{table*}
\centering
\caption{Comparison of our post-variational strategies to the conventional variational neural network.}
\label{tab:compare}
\begin{tabularx}{\textwidth}{lXX}
\hline\hline
& Variational~\citep{cerezo2021variational} & Post-variational\\
\hline
Circuit &  Single parameterized quantum circuit. & Multiple fixed quantum circuits.\\
\hline
Measurement observable & Single fixed observable. & Multiple fixed observables.\\
\hline
Optimization parameters & Parameters exist on the quantum computer and modify the quantum circuit. &  Parameters exist on the classical computer and modify the coefficients defining the weighted combination of quantum circuits/observables.\\
\hline
Feedback loop & Requires hybrid classical-quantum loop to modify parameters on quantum computer. & Measurements are executed in one go on quantum computer. Parameter optimization  done classically without usage of the quantum computer.\\
\hline
Optimization landscape & Non-convex optimization with a potentially exponential number of local minima far from global minima. & Convex quadratic optimization.\\
\hline
Number of measurements & Dependent on optimization landscape and initial parameters. & Dependent on number of quantum circuits. See Table~\ref{tab:strategy_bound}.\\
\hline
Challenges & Barren plateaus. & Heuristic choice of fixed circuits and observables from an exponential amount of possible circuits.\\
\hline\hline
\end{tabularx}
\end{table*}

\subsection{Classical combinations of quantum observables} 
The idea of using linear combinations of unitaries has been shown to be useful in both fault-tolerant~\citep{childs2012hamiltonian, berry2015simulating} and near-term settings~\citep{huang2021near}. In contrast, our goal is to take combinations of outputs of quantum circuits in our post-variational strategy. We generalize the idea of taking classical combinations of quantum states (CQS) to taking the \emph{classical combinations of quantum observables} (CQO) by combining the Ansatz $U(\theta)$ and observable $O$ into a single parameterized observable $\mathfrak{O}(\theta)$ and replacing this observable with a collection of predefined trial observables $\mathfrak{O}_1, \mathfrak{O}_2, \cdots, \mathfrak{O}_m$. Under this setting, measurement results on the quantum circuits are then combined classically, where the optimal weights of each measurement is computed via classical neural networks over a convex optimization problem. 

Starting from our variational observable $O$, we can combine the parameterized Ansatz with the observable to obtain a parameterized observable $\mathfrak{O}(\boldsymbol{\theta}) := U^\dagger(\boldsymbol{\theta})OU(\boldsymbol{\theta})$. The expectation values are the same as $\tr(O\varrho(\boldsymbol{\theta}, \boldsymbol{x})) = \tr(\mathfrak{O}(\boldsymbol{\theta}) \rho(\boldsymbol{x})).$  Therefore, instead of optimizing over parameterized trial quantum states, we can optimize over a parameterized observable to achieve the same effect. We note that this consideration is similar to the Heisenberg picture in contrast to the Schr\"odinger picture that variational algorithms take to construct trial wavefunctions. 

As any observable can be expressed as a linear combination of Hermitians, one can express the observable $\mathfrak{O}(\boldsymbol{\theta})$ as linear combinations weighted by functions $\mathcal{F}_j: \mathbb{R}^k \to \mathbb{R}$ of $\boldsymbol{\theta}$ such that 
\begin{equation}
\mathfrak{O}(\boldsymbol{\theta}) = \sum_{j=1}^{M} \mathcal{F}_j(\boldsymbol{\theta}) \mathfrak{O}_j,
\end{equation}
where the total number of decomposed observables $M$ is upper bounded by $4^n$. We leave the construction of such decompositions to Appendix~\ref{appendix:decompose}.

Noting that the estimator can be written as follows, $\mathcal{E}_{\boldsymbol{\theta}}(\boldsymbol{x}) = \tr(\mathfrak{O}(\boldsymbol{\theta})\rho(\boldsymbol{x})) = \sum_{j=1}^M\mathcal{F}_j(\boldsymbol{\theta})\tr(\mathfrak{O}_j\rho(\boldsymbol{x}))$, we can model such systems by considering the entire system as a function $\mathcal{H}_\theta : \mathbb{R}^M \to \mathbb{R}$ and the $M$ terms of $\tr(\mathfrak{O}_j\rho(\boldsymbol{x}))$ such that 
\begin{equation}
\mathcal{E}_{\boldsymbol{\theta}}(\boldsymbol{x}) = \mathcal{H}_{\boldsymbol{\theta}}(\tr({\mathfrak{O}_j}\rho(\boldsymbol{x}))\}_{j=1}^M).
\end{equation}
By universal approximation theorem~\citep{hornik1989multilayer}, we can approximate function $\mathcal{H}_{\boldsymbol{\theta}}$ classically through neural network model $\mathcal{G}_{\boldsymbol{\alpha}}$ parameterized by classical parameters $\boldsymbol{\alpha}$, such that 
\begin{equation}
\mathcal{E}_{\boldsymbol{\theta}}(\boldsymbol{x}) \approx \mathcal{G}_{\boldsymbol{\alpha}}(\{\tr({\mathfrak{O}_j}\rho(\boldsymbol{x}))\}_{j=1}^M) = \mathcal{E}_{\boldsymbol{\alpha}}(\boldsymbol{x}).
\end{equation}
Under this framework, the estimator can be further extended to simulate non-linear systems.

In simple linear cases, a concrete procedure for creating classical combinations of quantum observables is as follows. Make a first approximation and consider only a subset $\mathcal S$ of size $\vert\mathcal S\vert = m$ of the trial observables.  This leads to 
\begin{equation}
\mathfrak{O}(\boldsymbol{\theta}) \approx \sum_{j: \mathfrak{O}_j\in \mathcal S} \mathcal{F}_j(\boldsymbol{\theta})\mathfrak{O}_j.
\end{equation}
We then make our second approximation and consider all the functions independently. 
Given the optimal observable $\mathcal O(\boldsymbol{\theta^*})$, we consider the value of $\mathcal{F}_j(\boldsymbol{\theta^*})$ as a classical parameter $\alpha_j \in \mathbb R$. This leads to
\begin{equation}
\mathfrak{O}(\boldsymbol\theta^*) \to \mathfrak{O}(\boldsymbol{\alpha}) := \sum_{j: \mathfrak{O}_j\in \mathcal S}\alpha_j\mathfrak{O}_j.
\end{equation}
Hence, we obtain a linear combination of observables, which, assuming the approximations are chosen judiciously, will contain much of the initial complexity. 

\subsection{Comparisons with CQS} 
The CQS approach \citep{huang2021near} was originally proposed to solve linear systems of equations of $Ax=\boldsymbol{b}, A\in \mathcal{M}_{N\times N}(\mathbb C), \boldsymbol{b}\in \mathbb{C}^N$, where the solution of $x$ can be found by taking the pseudoinverse of $A$ such that $x=A^+\boldsymbol{b}$. With the variational approach, $A^+$ is modeled by a variational Ansatz $U(\theta)$, while the CQS approach takes classical combinations of fixed ``problem-inspired'' Ans\"atze $U_1, U_2, \cdots, U_{m_{\mathrm{CQS}}}$ generated by the matrix $A$ in conjunction with an Ansatz tree. An estimator $\hat{A}^+$ of $A^+$ is constructed such that $\hat{A}^+ = \sum_{i'=1}^{m_{\mathrm{CQS}}} \gamma_{i'}U_{i'}$, where $\gamma_i \in \mathbb{R}_{>0}$.

We note that the CQS approach can be viewed as a problem-inspired analogue to post-variational strategies when viewing the problem under the Hamiltonian loss. For any $\ket {\hat x}$, one can define the Hamiltonian loss as  $\mathcal{L}_{\mathrm{Ham}}(\ket {\hat x}) := \langle \hat{x} | O |\hat{x}\rangle$~\citep{bravo_prieto2020variational, huang2021near}, where $O= A^\dagger (\mathbb{I} - |b\rangle\langle b|) A$.
In particular, using the CQS Ansatz, we obtain $|\hat{x}\rangle = \hat{A}^+|b\rangle = \sum_{i'=1}^{m_{\mathrm{CQS}}} \gamma_{i'}U_{i'}|b\rangle$, and hence the loss function evaluates to
\begin{align}
&\mathcal{L}_{\text{Ham, CQS}}\nonumber\\
&= \tr\left( \left(\sum_{{i'}=1}^{m_{\mathrm{CQS}}}\gamma_{i'}U_{i'}^\dagger\right)O\left(\sum_{{j'}=1}^{m_{\mathrm{CQS}}} \gamma_{j'}U_{j'}\right) | b \rangle \langle b |\right)\\
&\begin{multlined}[t]=\sum_{{i'}=1}^{m_{\mathrm{CQS}}} \gamma_{i'}^2\tr(U_{i'}^\dagger O U_{i'}| b \rangle \langle b |) \\+ \sum_{{i'}=1}^{m_{\mathrm{CQS}}} \sum_{{j'}\ne {i'}} \gamma_{i'}\gamma_{j'}\tr\left(\frac{(U_{i'}^\dagger O U_{j'}) + (U_{j'}^\dagger O U_{i'})}{2}| b \rangle \langle b |\right).\end{multlined}
\end{align}
Collecting the terms and corresponding coefficients of $U_{i'}^\dagger O U_{i'}$ and $\frac{1}{2} ((U_{i'}^\dagger O U_{j'}) + (U_{j'}^\dagger O U_{i'}))$ into predefined trial observables $\mathfrak{O}_j$ and their corresponding coefficients $\alpha_j \in \mathbb{R}$, we find that we can write the Hamiltonian loss of the CQS approach as the mean-absolute error (MAE) loss of the post-variational approach as
\begin{align}
\mathcal{L}_{\text{Ham, CQS}} &= \sum_{j=1}^{m} \alpha_j\tr(\mathfrak{O}_j | b \rangle \langle b |)\\
&= \left|\sum_{i=1}^{m} \alpha_j\tr(\mathfrak{O}_j | b \rangle \langle b |) - 0\right| \\
&= \mathcal{L}_{\text{MAE, Post-Variational}}\\
&\le \mathcal{L}_{\text{RMSE, Post-Variational}}
\end{align}
where we note that 0 serves as the ground truth of the MAE loss as $\langle x| A^\dagger(\mathbb{I} - |b\rangle\langle b|)A |x\rangle =0$. Here, $m = m_{\mathrm{CQS}}^2$ by counting the different terms. Note that we can consider the absolute value directly as $\mathbb{I} - |b\rangle\langle b|$ is a projector, indicating that its eigenvalues are 0 or 1, which further implies that $0 \le \bra{\hat x}A^\dagger (\mathbb{I} - |b\rangle\langle b|) A\ket{\hat x} =\bra{\hat x}O\ket{\hat x} = \sum_{j=1}^{m} \alpha_j\tr(\mathfrak{O}_j | b \rangle \langle b |) $.

\begin{figure*}
    \centering
    \includegraphics[width=\linewidth]{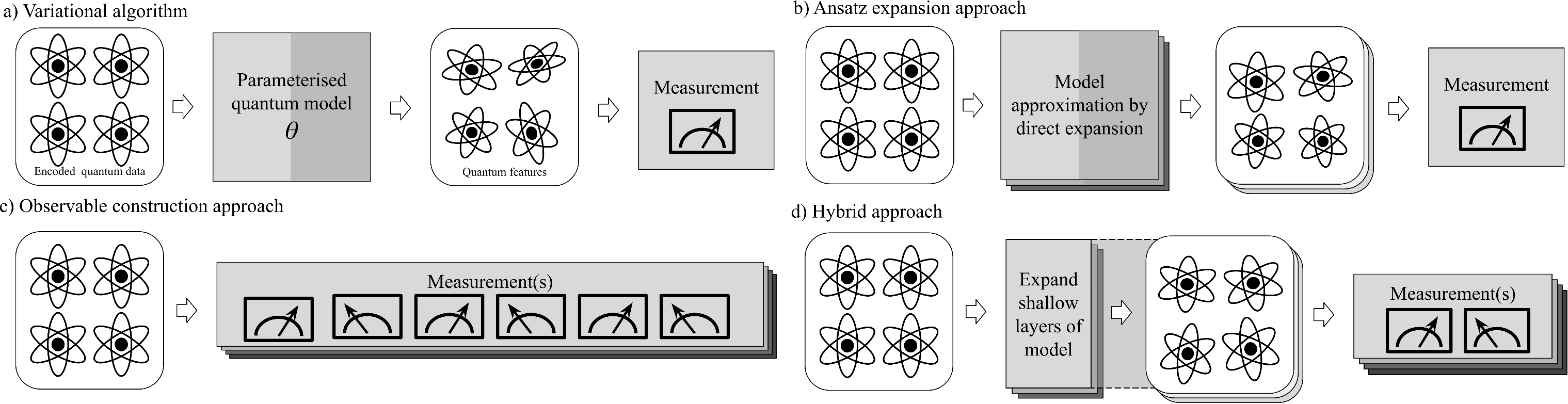}
    \caption{Here we provide an overview of the various strategies introduced in the later text. Using the variational circuit (a) as a baseline, the Ansatz expansion approach (b) does model approximation by directly expanding the parameterized Ansatz into an ensemble of fixed Ans\"atze. On the other hand, the observable construction approach (c) foregoes all usage of an Ansatz and directly constructs a measurement observable by ensembling and taking the classical combination of various predefined trial observables. The hybrid approach (d) does both the expansion of the Ansatz, albeit only on the shallower layer of the model, and replaces the deeper layers of the model as well as the measurement observable with a series of measurement trial observables that can be used to retrieve a classical combination.}
    \label{fig:strategy}
\end{figure*}

\section{Design principles of post-variational quantum circuits}
\label{ch:strategy}
We now discuss multiple possible heuristic strategies to decide on the trial observables $\mathfrak{O}_j$ and construct the circuits for our post-variational algorithm to minimize operations on quantum computers. See Figure~\ref{fig:strategy} for an overview of the heuristic strategies. Recall that we construct a post-variational algorithm by ensembling multiple trial observables such that the final target observable can be learned by combination of the measurement results based on a learned function:
\begin{equation}
 \mathcal{E}_{\boldsymbol{\theta}}(\boldsymbol{x}) \approx \mathcal{E}_{\boldsymbol{\alpha}}(\boldsymbol{x}) =\mathcal{G}_{\boldsymbol{\alpha}}(\{\tr({\mathfrak{O}_j}\rho(\boldsymbol{x}))\}_{j:\mathfrak{O}_j\in \mathcal{S}}),
\end{equation}
where for linear cases
\begin{equation}
\mathcal{G}_{\boldsymbol{\alpha}}(\{\tr(\mathfrak{O}_j\rho(x))\}_{j:\mathfrak{O}_j\in \mathcal{S}}) = \sum_{j:\mathfrak{O}_j \in \mathcal{S}} \alpha_j \tr(\mathfrak{O}_j\rho(\boldsymbol{x})).
\end{equation}

\begin{figure}
\includegraphics[width=0.9\linewidth]{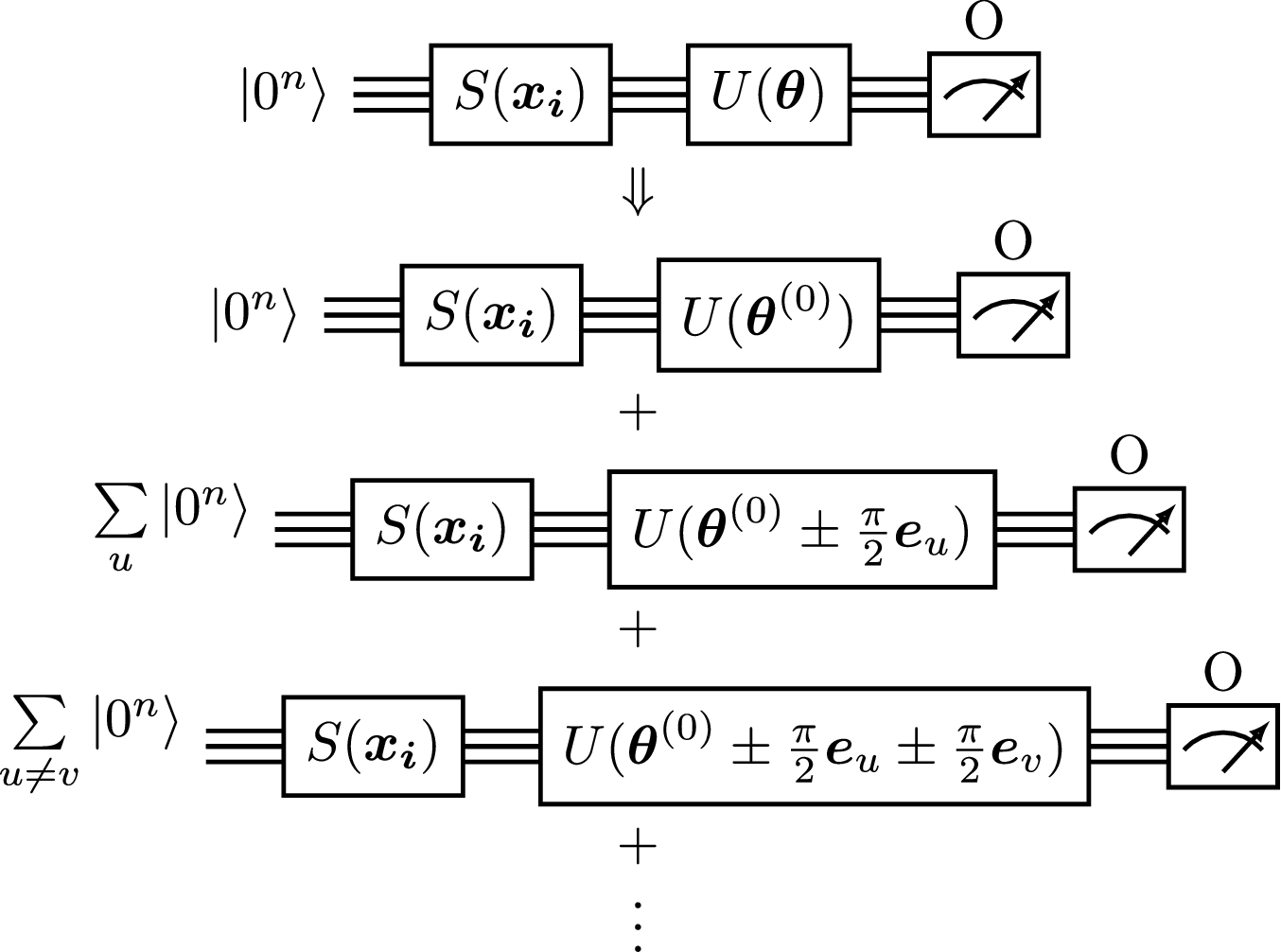}
\caption{Ansatz expansion approach to post-variational algorithm construction. Starting from a variational Ansatz, multiple non-parameterized quantum circuits are constructed by Taylor expansion of the Ansatz around a suitably chosen initial setting of the parameters $\boldsymbol{\theta^{(0)}}$. Gradients and higher-order derivatives of circuits can be obtained by parameter-shift rule. The different circuits are linearly combined with classical coefficients that are optimized via convex optimization. }
\label{fig:ansatz_expand}
\end{figure}

\subsection{Ansatz expansion} The first strategy to constructing post-variational algorithms is to first begin with a variational algorithm and replace the parameterized Ansatz $U(\boldsymbol{\theta})$ with an ensemble of $p$ fixed Ans\"atze $\{U_a\}_{a=1}^p$. \citet{huang2021near}'s CQS approach uses this strategy to solve linear systems, generating the fixed problem-inspired  Ans\"atze through the usage of an Ansatz tree. 

For quantum machine learning tasks, such problem-inspired Ans\"atze may not be readily available, hence we generate the fixed Ans\"atze from a problem-agnostic Ansatz through some expansion of the Ansatz.

Works in dequantizing parameterized quantum circuits~\citep{schreiber2023classical,sweke2023potential, landman2023classically} show that some parameterized quantum circuits that use data reuploading models~\citep{perezsalinas2020data} can be expressed with some error tolerance by classical models through the usage of partial Fourier series~\citep{gilvidal2020input,schuld2021effect}. However, such expansions are conditioned on specific data encoding circuits and further, the ``classical surrogate'' models does not retain the feature extraction abilities of the original Ansatz as the information of the Ansatz is compacted into the Fourier coefficients of the surrogate model. Our usage of truncated Taylor series as a method to generate post-variational circuits retains the such abilities of the original Ansatz circuit.

Inspired by the idea of using truncated Taylor series as a linear combination of unitaries~\citep{berry2015simulating}, we use truncated Taylor polynomial expansions of the variational parameters in such quantum circuits to generate the fixed Ans\"atze of our model. We use parameter-shift rules~\citep{mitarai2018quantum, schuld2019evaluating} to find derivatives of the trace-induced measurements of parameterized quantum circuits, as well as extending to derivatives of higher-orders~\citep{huembeli2021characterizing, mari2021estimating}. 

Note that as parameter-shift rules provide the gradient of the entire quantum circuit instead of the variational Ansatz. We in fact decompose the variational \emph{observable} $\mathfrak O(\boldsymbol{\theta}) = U^\dagger(\boldsymbol{\theta})OU(\boldsymbol{\theta})$ into a truncated Taylor series, rather than considering the linear combination of multiple Ans\"atze generated from derivatives of the Ansatz.
This approach falls under the framework of using CQO, instead of directly decomposing the Ansatz itself.

With the parameter-shift rule, one can compute the gradient of the variational observable $\mathfrak O(\boldsymbol{\theta})$ by a linear combination of same circuit, but with the parameters shifted up and down from the initial value by a specific value dependent on the gate~\citep{wierichs2022general}. In the post-variational setting, as we have to obtain the circuits whose combination can be used to find the gradient instead of the exact combination that retrieves the gradients and higher-order gradients, one can simply decompose the parameterized gates into a series of rotational Pauli gates to obtain simpler parameter-shift rules, which would then be combined via derivative product rules when under the variational setting~\citep{crooks2019gradients}. The parameter-shift rule for rotational Pauli gates show that to compute gradients, one simply needs to shift that specific parameter up and down by $\frac{\pi}{2}$~\citep{mitarai2018quantum}. Further, \citet{mari2021estimating} has shown that higher-order gradients of a single parameter can be computed via a linear combination of circuits that set the parameter to the values of $\{0, \pm\frac{\pi}{2}\}$ and that through full Taylor expansion of the circuit, one can express $U^\dagger(\boldsymbol{\theta})OU(\boldsymbol{\theta})$ of any arbitrary $\boldsymbol\theta\in \mathbb{R}^k$ as a linear combination of $U^\dagger(\boldsymbol{\theta'})OU(\boldsymbol{\theta'})$ where $\boldsymbol{\theta'} \in \{0, \pm\frac{\pi}{2}\}^k$.

For our Ansatz expansion strategy, we assume the simple case of initializing all parameters to 0 as rotational gates evaluate to identity at parameter 0 and as a result, we can shorten the depth of the circuit during execution, but note that random initialization also would work. Truncating at the $R$-th derivative order, to generate the additional circuits, we simply select all combinations of size $\le L$ from the $k$ parameters in $\theta$, where each parameter corresponds to a single rotational gate, and set each parameter to $\pm\frac{\pi}{2}$. See Figure~\ref{fig:ansatz_expand} for a representation of such quantum circuits. One can then see that for truncation at derivative order $R$, the number of circuits required is then
\begin{equation}
\sum_{\ell=0}^R \binom{k}{\ell}2^\ell \in \mathcal O(2^{R}k^{R}).
\end{equation}

We note that the number of circuits required including the full truncated Taylor series is taken to a power of the original number parameters, which can scale largely when a deep Ansatz is used. To reduce the number of circuits required, we can directly prune the circuits that generate small circuits. Suppose we have a circuit that has a variational observable $U^\dagger(\boldsymbol\theta)OU(\boldsymbol\theta)$ that we want to take the gradient of on the $u$-th parameter $\theta_u$ using parameter shift rules. We then construct two variational observables $U^\dagger(\boldsymbol\theta+\frac{\pi}{2}\boldsymbol{e_u})OU(\boldsymbol\theta+\frac{\pi}{2}\boldsymbol{e_u})$ and $U^\dagger(\boldsymbol\theta-\frac{\pi}{2}\boldsymbol{e_u})OU(\boldsymbol\theta-\frac{\pi}{2}\boldsymbol{e_u})$. Now considering the data, if the MSE over 
\begin{equation}
\label{eqAnsatzDiff}
\begin{multlined}[b]
\tr(U^\dagger(\boldsymbol\theta+\tfrac{\pi}{2}\boldsymbol{e_u})OU(\boldsymbol\theta+\tfrac{\pi}{2}\boldsymbol{e_u})\rho(\boldsymbol{x_i})) \\
- \tr(U^\dagger(\boldsymbol\theta-\tfrac{\pi}{2}\boldsymbol{e_u})OU(\boldsymbol\theta-\tfrac{\pi}{2}\boldsymbol{e_u})\rho(\boldsymbol{x_i}))
\end{multlined}
\end{equation}
is small, then the gradient on $\theta_u$ of $U^\dagger(\boldsymbol\theta)OU(\boldsymbol\theta)$ is also small, and we can prune one of the circuits (or even both) from the list of circuits as further higher-order gradients based on the gradient circuits would also be small.

\begin{figure}
\includegraphics[width=0.6\linewidth]{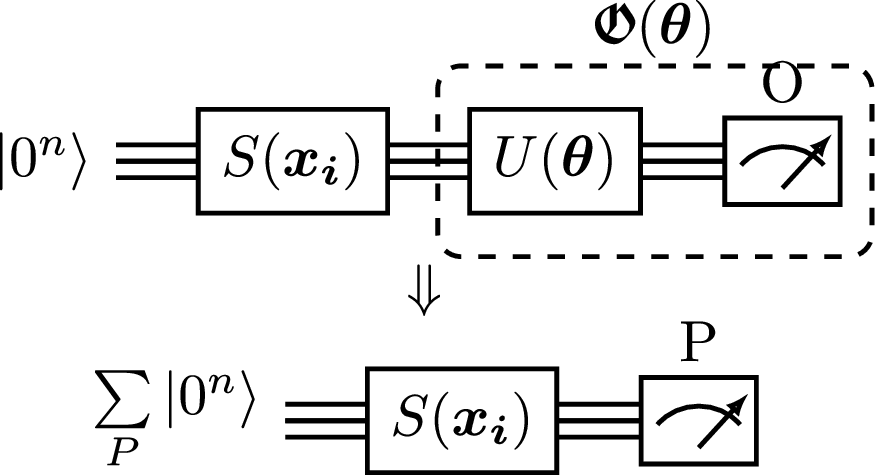}
\caption{Observable construction approach to post-variational algorithm construction. A variational observable can be directly constructed by an ensemble of Pauli observables, and may serve as a potential approximation under locality restrictions.}
\label{fig:observable_construct}
\end{figure}

\subsection{Observable construction} In contrast to the Ansatz expansion strategy, where we generate fixed Ans\"atze either through Ansatz trees or Taylor expansions, in the observable construction strategy, we take the CQO strategy discussed in the previous section at face value, decomposing the parameterized observable $\mathfrak{O}(\boldsymbol{\theta})$ against the basis of quantum observables (namely, Paulis), such that $\mathfrak{O}(\boldsymbol{\theta^*}) \to \mathfrak{O}(\boldsymbol{\alpha}) = \sum_{P\in \{\mathcal{I}, \mathcal{X}, \mathcal{Y}, \mathcal{Z}\}^{\otimes n}} \alpha_P P$.  See Figure~\ref{fig:observable_construct} for a representation of such quantum circuits.

However, such a strategy scales exponentially with the number of qubits used in the system, creating an analogue to the barren plateau. Further heuristic selections of observables would be required to prevent such exponential scaling. We find that considering all Pauli observables within a certain locality $L$ as a good heuristic given that most physical Hamiltonians are local. Furthermore, \citet{huang2023learning} show that this truncation by locality of Pauli observables, which they refer to as ``low-degree approximations'', is proven to be a good surrogate for the unknown observable-to-be-learned under the circumstance that the quantum state $\rho$ that the observable is applied to is sampled from distributions that are invariant under single-Clifford gates. They further show that these low-degree approximations would require only a quasi-polynomial number of samples in relation to the number of qubits $n$ to learn.

Under the circumstance that the target observable $\mathfrak{O}$ is $L$-local, we can employ the classical shadows protocol~\citep{huang2020predicting} to reduce the quantum measurements needed on the quantum computer while achieving the same additive error term $\epsilon$ determined by the loss term. For each data sample $x$ such that we can prepare a quantum state $\rho(\boldsymbol{x})$, we can obtain a series of classical shadows of the quantum state $\hat{\rho}(\boldsymbol{x})$ that can be stored classically such that we can estimate the values of $\tr{(P\rho(\boldsymbol{x}))}$ where $P \in \{\mathcal{I}, \mathcal{X}, \mathcal{Y}, \mathcal{Z}\}^{\otimes n}:|P|\le L$. The classical shadows protocol is able to reduce the number of measurements required for all values for $\tr(P\rho(\boldsymbol{x}))$ from a polynomial dependency of the number of qubits to a logarithmic dependency. 

Consider the case where the observables are the complete set of $L$-local Paulis. The number of observables $q$ is then 
\begin{equation}
\sum_{\ell=0}^L \binom{n}{\ell}3^\ell \in \mathcal O(3^{L}n^{L}).
\end{equation}
If we use the classical shadows method, then we can reduce the number of random measurements of the circuit to 
\begin{equation}
\mathcal O(3^L\log q) \in \mathcal O(3^LL\log n)
\end{equation}

\begin{figure}
\includegraphics[width=0.9\linewidth]{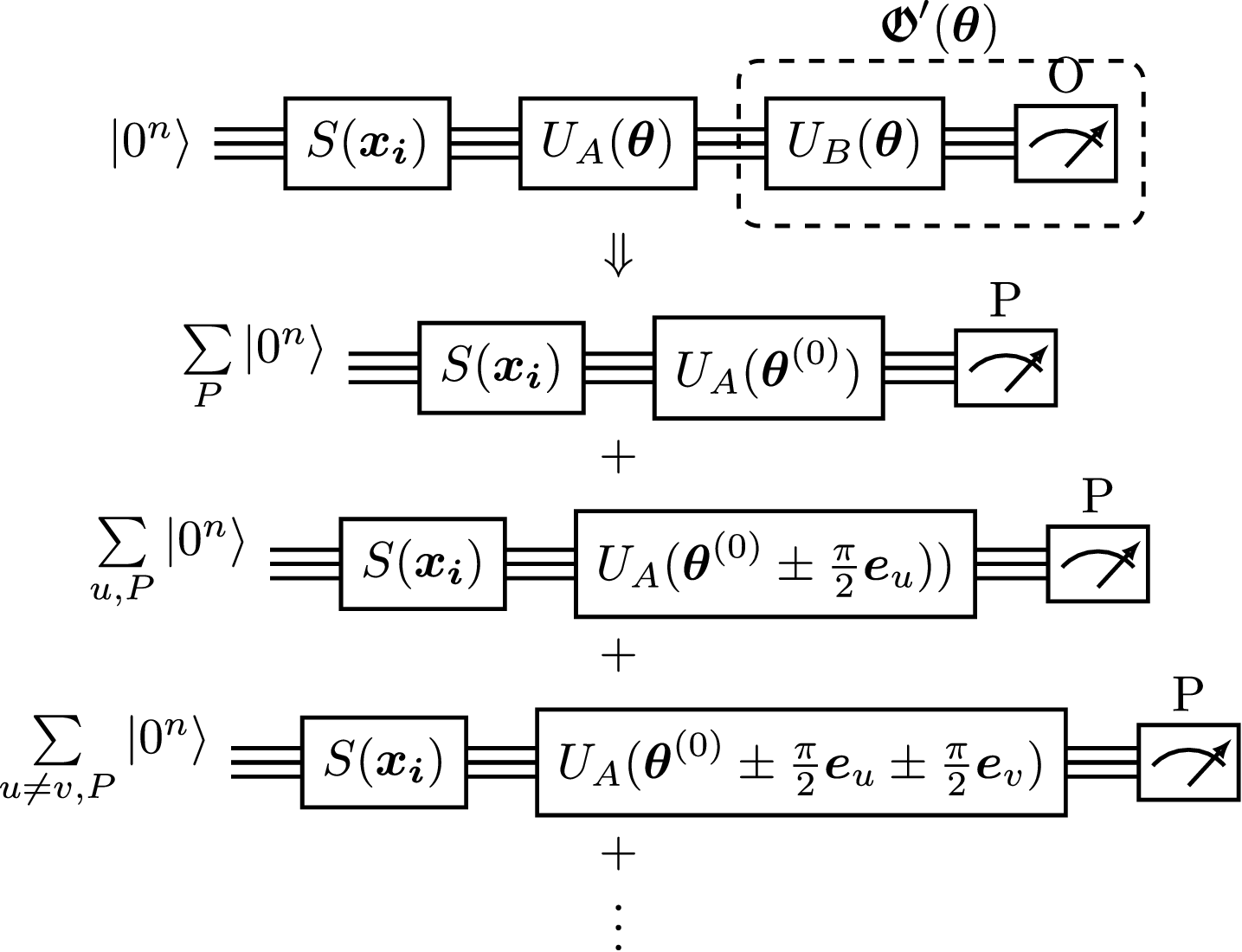}
\caption{Hybrid approach to post-variational algorithm construction. This approach is a combination of the approaches shown in Figures~\ref{fig:ansatz_expand} and \ref{fig:observable_construct}.}
\label{fig:hybrid}
\end{figure}

\subsection{Hybrid approach} 
When taking the strategy of observable construction, one additionally may want to use Ansatz quantum circuits to increase the complexity of the model. Hence, we discuss a simple hybrid strategy that combines both the usage of Ansatz expansion and observable construction. See Figure~\ref{fig:hybrid} for a representation of such quantum circuits.

Recall the definition of the parameterized observable as $\mathfrak{O}(\boldsymbol{\theta}) = U^\dagger(\boldsymbol{\theta})OU(\boldsymbol{\theta})$. Instead of expanding $U(\boldsymbol{\theta})$ directly, we split the Ansatz $U(\boldsymbol{\theta})$ into two unitaries (based on cutting the circuit at a certain depth, for example), such that 
\begin{equation}
U(\boldsymbol{\theta}) = U_B(\boldsymbol{\theta})U_A(\boldsymbol{\theta}).
\end{equation}
We can now write $\mathfrak{O}(\boldsymbol{\theta}) = U_A^\dagger(\boldsymbol{\theta})U_B^\dagger(\boldsymbol{\theta})OU_B(\boldsymbol{\theta})U_A(\boldsymbol{\theta}).$
We denote $\mathfrak{O}'(\boldsymbol{\theta}) := U_B^\dagger(\boldsymbol{\theta})OU_B(\boldsymbol{\theta})$.

We first decompose $\mathfrak O'(\boldsymbol{\theta})$ directly into a linear combination of Paulis via observable construction approach. Next, we expand the remaining Ansatz $U_A(\boldsymbol{\theta})$ using the Ansatz expansion method. Our next step is to prune the number of circuits to a reasonable amount.

From the Ansatz expansion method, we can obtain from the initial circuit, various derivative circuits that make up the derivatives of the initial circuit. Combining the derivative circuits with a global observable results in a derivative of global cost functions, which have been shown to potentially produce small derivatives and encounter barren plateaus by \citet{cerezo2021cost}. Hence, we can set a locality $L$ such that all derivatives circuits that have observables larger that $L$ are pruned from our list of circuits. 

Consider the subset of circuits that have observables within locality $L$ from our pruned list of circuits. If we set the initial parameters such that Ansatz evaluates to identity, then the initial circuits have at least the expressibility of the observable construction method that has the imposed constraint on locality. The derivative circuits with local observables can then simulate to a limited extent observables that lay outside of the restricted locality $L$ that we constrain our observables to. This provides additional expressibility without having to measure on all possible Paulis. Given this heuristical boost in expressibility stemming from the derivative circuits, we retain only the subset of circuits with observables within locality $L$ to prevent an exponential amount of circuits. With the observables used restricted to local observables, we can once again employ the classical shadows method to reduce measurements.

Recall that in the Ansatz expansion method, we prune small gradients by finding the difference of measured results of the parameter-shifted circuits against an observable $O$. In the hybrid method, we can further conduct this type of pruning. As we are dealing with multiple observables, instead of finding the difference between the measurements, one can also directly measure the fidelity between the states $U(\boldsymbol\theta+\frac{\pi}{2}\boldsymbol{e_u})\rho(\boldsymbol{x_i})U^\dagger(\boldsymbol\theta+\frac{\pi}{2}\boldsymbol{e_u})$ and $U(\boldsymbol\theta-\frac{\pi}{2}\boldsymbol{e_u})\rho(\boldsymbol{x_i})U^\dagger(\boldsymbol\theta-\frac{\pi}{2}\boldsymbol{e_u})$. This can be justified by upper bounding the MSE of Equation~\ref{eqAnsatzDiff} with the trace distance, as follows.

Let $\varrho(\boldsymbol{x}, \boldsymbol{\theta}) = U(\boldsymbol\theta)\rho(\boldsymbol{x})U^\dagger(\boldsymbol\theta)$. Then Equation~\ref{eqAnsatzDiff} can be rewritten as 
\begin{equation}
\begin{multlined}[b]
\tr(P\varrho(\boldsymbol{x}, \boldsymbol{\theta}+\tfrac{\pi}{2}\boldsymbol{e_u})) - \tr(P\varrho(\boldsymbol{x}, \boldsymbol{\theta}-\tfrac{\pi}{2}\boldsymbol{e_u}))\\
= \tr(P(\varrho(\boldsymbol{x}, \boldsymbol{\theta}+\tfrac{\pi}{2}\boldsymbol{e_u})-\varrho(\boldsymbol{x}, \boldsymbol{\theta}-\tfrac{\pi}{2}\boldsymbol{e_u}))).
\end{multlined}
\end{equation}
The MSE of over all data points can then be written as 
\begin{equation}
\frac{1}{|\mathcal D|}\sum_{x_i\in\mathcal{D}} \left|\tr(P(\varrho(\boldsymbol{x}, \boldsymbol{\theta}+\tfrac{\pi}{2}\boldsymbol{e_u})-\varrho(\boldsymbol{x}, \boldsymbol{\theta}-\tfrac{\pi}{2}\boldsymbol{e_u})))\right|^2.
\end{equation}
From the definition of trace distance, we note that 
\begin{align}
&\tr(P(\varrho(\boldsymbol{x}, \boldsymbol{\theta}+\tfrac{\pi}{2}\boldsymbol{e_u})-\varrho(\boldsymbol{x}, \boldsymbol{\theta}-\tfrac{\pi}{2}\boldsymbol{e_u})))^2 \\
&\le  4\left(\tr|\varrho(\boldsymbol{x}, \boldsymbol{\theta}+\tfrac{\pi}{2}\boldsymbol{e_u})-\varrho(\boldsymbol{x}, \boldsymbol{\theta}-\tfrac{\pi}{2}\boldsymbol{e_u})|\right)^2\\
&\le 4\left(1 - F\left(\varrho(\boldsymbol{x}, \boldsymbol{\theta}+\tfrac{\pi}{2}\boldsymbol{e_u}), \varrho(\boldsymbol{x}, \boldsymbol{\theta}-\tfrac{\pi}{2}\boldsymbol{e_u})\right)\right)
\end{align}
One can then see that if the fidelity of the shifted Anst\"aze are large, we can directly prune all the quantum circuits related to the shifted Ans\"atze. The pruning strategy in the Ansatz expansion method still works in this setting without needing to evaluate the result on all local observables. Given that $\rho(\boldsymbol{x_i})$ is a pure state in our case, the fidelity can be computed via the probability of the outcome of the quantum state $S^\dagger(\boldsymbol{x_i})U^\dagger(\boldsymbol\theta+\frac{\pi}{2}\boldsymbol{e_u})U(\boldsymbol\theta-\frac{\pi}{2}\boldsymbol{e_u})S(\boldsymbol{x_i})\ket{0^n}$ being the bit string $0^n$~\citep{havlivcek2019supervised}. One can conduct further pruning based on the observables after measurement to further reduce the amount of circuits.

\section{Architecture design of post-variational neural networks}
\label{ch:architecture}
Let it be possible to approximate the target observable $\mathfrak{O}(\boldsymbol{\theta})$ obtained in a variational algorithm with a predefined collection of trial observables $\mathcal{S} = \{\mathfrak{O}_1, \mathfrak{O}_2, \cdots \mathfrak{O}_m\}$, such that $\mathcal{E}_{\boldsymbol{\theta}}(\boldsymbol{x}) = \tr (\mathfrak{O}(\boldsymbol{\theta}) \rho(\boldsymbol{x})) \approx \sum_{j: \mathfrak{O}_j\in \mathcal S} \alpha_j\tr(\mathfrak{O}_j\rho(\boldsymbol{x}))$. Then the optimal values of the coefficients $\alpha_1, \alpha_2, \cdots, \alpha_m$ are found with classical optimization by perceptrons or neural networks. Following the idea of neural networks, we call each individual quantum circuit/observable pair as a quantum neuron~\citep{cao2017quantum}.

\begin{definition} \label{defHybrid}
We define a \emph{\texttt{(p, q)}-hybrid strategy} as a strategy with a total of $p$ Ans\"atze $\{U_a\}_{a=1}^p$ and $q$ observables $\{O_b\}_{b=1}^q$, such that the output of each circuit for any input $\rho$ is $\tr(\mathfrak{O}_{(a,b)} \rho) = \tr(U_a^\dagger O_b U_a \rho)$. 
\end{definition}
As the outcome of quantum neuron measurements are probabilistic, multiple measurements need to be conducted in order to produce a good estimate. 
\begin{figure}
\includegraphics[width=0.6\linewidth]{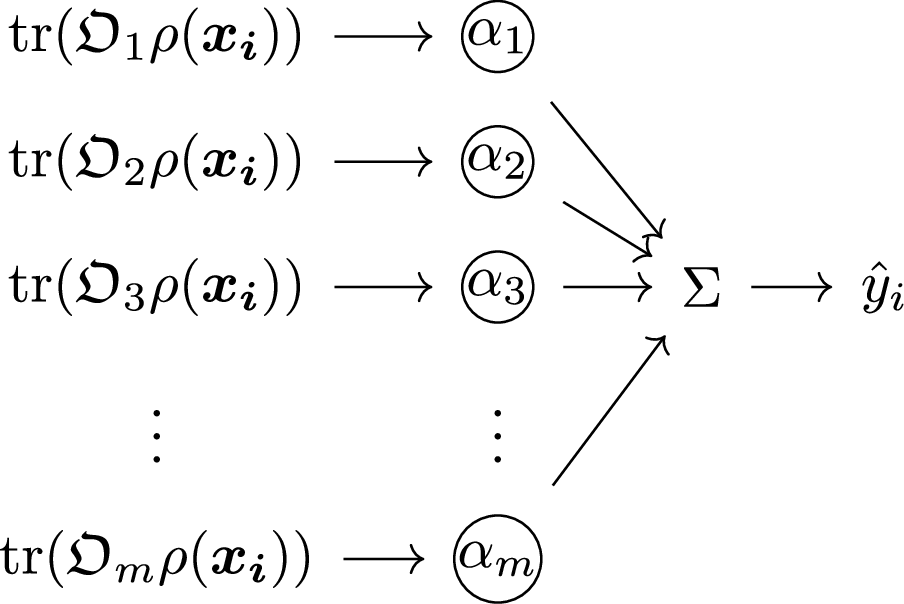}
\caption{Classical linear regression approach of the hybrid model for a \texttt{(p, q)}-hybrid strategy. Let $m=pq$. Given a set of observables $\mathfrak O_{j}$, $j \in [m]$, and data embedding $\rho(\boldsymbol{x_i})$, a simple linear quantum neural network can be constructed by combining the outputs in a linear fashion using combination parameters $\alpha_{j}$. This obtains the output label $\hat y_i$.}
\label{fig:lreg}
\end{figure}
We define matrix $Q\in \mathcal M_{d\times m}(\mathbb{R})$ such that 
\begin{equation}
Q_{ij} := \tr(\mathfrak{O}_j \rho(\boldsymbol{x_i})), \label{eq:q_def}
\end{equation}
where $\{\mathfrak{O}_j\}_{j=1}^m$ is the collection of observables produced by a \texttt{(p, q)}-hybrid strategy and $m=pq$. We use $Q$ as the input to a classical linear regression model (Figure~\ref{fig:lreg}). We minimize the root-mean-square error as follows
\begin{align}
\mathcal{L}_{\mathrm{RMSE}} &= \sqrt{\frac{1}{d} \sum_{i=1}^d \left (y_i -\sum_{j=1}^{m} \alpha_{j}\tr(\mathfrak{O}_{j} \rho(\boldsymbol{x_i}))\right)^2}\\
&= \frac{1}{\sqrt d} \sqrt{\sum_{i=1}^d\left(y_i-\sum_{j=1}^m\alpha_jQ_{ij}\right)^2} \\
&= \frac{1}{\sqrt d} \|\boldsymbol{Y}-Q\boldsymbol{\alpha}\|_2,\label{eq:rmse}
\end{align}
where $\boldsymbol{Y} = (y_1, y_2, \cdots, y_d)^\intercal$. The closed form solution to solving linear regression problems obtains $\boldsymbol{\alpha} = Q^+\boldsymbol{Y}$, or in the special case where $Q$ has full column rank and thus $Q^\intercal Q$ is positive definite, we obtain $\boldsymbol{\alpha} = (Q^\intercal Q)^{-1}Q^\intercal \boldsymbol{Y}$.

Here, we use a simple linear regression model as the main classical model. However, this model can be extended to any classical model including feed-forward models; we only mention linear regression models in detail for easy analysis with closed-form solutions. Furthermore, the results can be extended to classification problems by adding an extra sigmoid or softmax function at the end of the output. 

Structurally, the post-variational quantum neural network mimics a two-layer classical multilayer perceptrons that has frozen parameters but multiple activations for the first layer. The fixed unitary Ans\"atze correspond to fixed linear feature extractors similar to the first layer of multilayer perceptrons, while the measurement correspond to nonlinear activation functions~\citep{schuld2020circuit}. As we measure on different observables, this can be viewed akin to using different types of non-linear activation functions. Lastly, the linear layer in our framework acts exactly like that of the second linear layer in classical two-layer neural networks.

\section{Error analysis of post-variational neural networks}
\label{ch:erroranalysis}
Given the design of post-variational neural networks based on a simple linear regression model as the classical ensemble model for quantum neurons.
Measurement outcomes in quantum systems are probabilistic, and hence estimated quantities come with errors that may be carried and propagated throughout the entire neural network system. In this section, we discuss the effects of such errors as well as the number of measurements required to achieve a target final error.

\subsection{Estimation errors} We first show the simple result for the total number of measurement for all terms $\tr(\mathfrak{O}_{j} \rho(\boldsymbol{x_i}))$.
\begin{proposition}[Measurements needed for direct estimation of all quantum neurons]\label{proposition:measure}
Consider a \texttt{(p, q)}-hybrid strategy as in Definition \ref{defHybrid}. Using the sample mean over multiple iterations as an estimator to evaluate the output of each of the $m=pq$ quantum neurons, the number of quantum measurements required to estimate the output of all quantum neurons over all $d$ data points such that the all $m\times d$ outputs have an additive error of $\epsilon_H$ with a probability of $1-\delta$ falls within 
$$\mathcal O\left(\frac{md}{\epsilon_H^2}\log \frac{md}{\delta}\right).$$
\end{proposition}
The proof is an application of Hoeffding~\citep{hoeffding1963probability} and union bound~\citep{boole2009mathematical} and is shown in Appendix \ref{appMeasure}.
Alternatively, one can reduce the number of total measurements by estimating the output of all quantum neurons that contain the same fixed Ansatz over a single data point by utilizing the classical shadows protocol.

\begin{proposition}[Measurements needed for shadow estimation of all quantum neurons]\label{proposition:shadow}
Consider a \texttt{(p, q)}-hybrid strategy as in Definition \ref{defHybrid}. Using the classical shadows to estimate output of all $m=pq$ quantum neuron over all $d$ data points, the number of quantum measurements required such that the all $m\times d$ outputs have an additive error of $\epsilon_H$ with a probability of $1-\delta$ falls within 
$$\mathcal O\left(\frac{pd}{\epsilon_H^2}\max_{k\in[q]}\|O_k\|^2_S\log\frac{md}{\delta}\right).$$
\end{proposition}
The proof is via application of the bounds of the median-of-means estimator~\citep{nemirovsky1983problem} and union bound~\citep{boole2009mathematical} and is shown in Appendix \ref{appMeasure}. Relating back to the design principles, we note that the classical shadows protocol does not help with the Ansatz expansion approach as $q=1$ and $m=p$, and instead increases the complexity by $\|O\|_S^2$. 

\subsection{Error propagation through neural networks} 
Recall that we define $Q$ such that $Q_{ij} = \tr(\mathfrak{O}_j \rho(\boldsymbol{x_i}))$ from Equation~\ref{eq:q_def}. Given that each element in $Q$ is the expected outcome of a quantum neuron, we construct $\hat{Q}$ such that $\hat{Q}_{ij}$ is the estimated outcome of the quantum neuron that falls within an additive error of $\epsilon_H$ by multiple direct measurements in Proposition~\ref{proposition:measure} or by utilizing classical shadow estimations in Proposition~\ref{proposition:shadow}. Formally, $\forall i, j\quad |\hat{Q}_{ij} - Q_{ij}| \le \epsilon_H$. 
We now consider the effect of such errors on the loss function $\mathcal{L}$ that we use for machine learning problems.

For linear regression, we consider the  RMSE loss of $\mathcal{L}_{\mathrm{RMSE}}(\boldsymbol{\alpha},Q)$ from Equation~\ref{eq:rmse}.
We then define the following optimal parameters $\boldsymbol{\alpha^*}$ and $\boldsymbol{\hat{\alpha}^*}$, for $Q$ and $\hat{Q}$, respectively, where
\begin{align}
\boldsymbol{\alpha^*} &:= {\argmin}_{\boldsymbol{\alpha}} \mathcal{L}_{\mathrm{RMSE}}(\boldsymbol{\alpha},Q),\label{eq:alpha}\\
\boldsymbol{\hat{\alpha}^*} &:= {\argmin}_{\boldsymbol{\hat{\alpha}}} \mathcal{L}_{\mathrm{RMSE}}(\boldsymbol{\hat \alpha},\label{eq:hatalpha}\hat{Q}).
\end{align}
Further, we define using only $Q$,
\begin{equation}
\Delta \mathcal{L}_{\mathrm{RMSE}} := \mathcal{L}_{\mathrm{RMSE}}(\boldsymbol{\hat{\alpha}^*},Q) - \mathcal{L}_{\mathrm{RMSE}}(\boldsymbol{\alpha^*},Q).\label{eq:rmse_diff}
\end{equation}
We then obtain the following theorem, the proof of which can be found in Appendix~\ref{appError}.
\begin{theorem}[Linear regression theoretical guarantee]\label{theorem:linear}
Consider a \texttt{(p, q)}-hybrid strategy as in Definition \ref{defHybrid}, and let $m=pq$. Let $Q$ be the matrix as defined in Equation~\ref{eq:q_def}. Let $\epsilon >0$ and let $\hat{Q}$ be such that
\begin{multline*}
\|\hat{Q}-Q\|_{\max} < \min \bigg(\frac{\min(\sigma_{\min}(Q), \sigma_{\min}(\hat{Q}))}{\sqrt{\min(m, d) md}},\\
\frac{\epsilon }{6 \sqrt{m}\|\boldsymbol{Y}\|_2 \|Q\|\|Q^+\|^2 } \bigg). \end{multline*}
Then, for the loss difference as defined in Equation~\ref{eq:rmse_diff}, $\Delta\mathcal{L}_{\mathrm{RMSE}} < \epsilon$.
\end{theorem}
To find the number of measurements needed to obtain a loss difference within $\epsilon$, we first set the provided upper bounds to be within $\epsilon_H$ as found in Propositions~\ref{proposition:measure} and \ref{proposition:shadow}. Assuming a scenario when $\|Q\|\|Q^+\| =\frac{\|Q\|}{\sigma_{\min}(Q)} = \kappa_Q \in O(1)$, $\|\boldsymbol{Y}\|_2 \in O(\sqrt{d})$, $\|Q\| \in \Omega(\sqrt{d})$, we upper bound $\epsilon_H^{-1}$ with respect to $m$, $d$ and $\epsilon$ as follows:
\begin{align}
&\epsilon_H \ge \min\left({\tfrac{\min(\sigma_{\min}(Q), \sigma_{\min}(\hat{Q}))}{\sqrt{\min(m, d) md}}, \tfrac{\epsilon }{6 \sqrt{m}\|\boldsymbol{Y}\|_2 \|Q\|\|Q^+\|^2}}\right)\\
\Rightarrow \, & \frac{1}{\epsilon_H} \le \max\left({\tfrac{\sqrt{\min(m, d) md}}{\min(\sigma_{\min}(Q), \sigma_{\min}(\hat{Q}))}, \tfrac{6 \sqrt{m}\|\boldsymbol{Y}\|_2 \|Q\|\|Q^+\|^2}{\epsilon}}\right)\\
\Rightarrow \, & \frac{1}{\epsilon_H} \in \max\left(\mathcal O(m), \mathcal O\left(\frac{\sqrt{m}}{\epsilon}\right)\right) \in \mathcal O\left(\frac{m}{\epsilon}\right)
\end{align}

Combining the results in Propositions~\ref{proposition:measure} and \ref{proposition:shadow}, if we want the final loss term of our hybrid quantum-classical model to have an error within $\epsilon$, then the total number of measurements needed falls within the complexity of 
\begin{equation}
    t \in \mathcal O\left(\frac{m^3d}{\epsilon^2}\log \frac{md}{\delta}\right)
 \end{equation}
for the probability of $1-\delta$ using direct measurements, and 
\begin{equation}
    t \in \mathcal O\left(\frac{m^2pd\max_{k\in[q]}\|O_k\|^2_{S}}{\epsilon^2}\log \frac{md}{\delta}\right)
\end{equation}
using the classical shadows protocol.

\begin{table*}
\centering
\caption{Upper bounds for the number of measurements for the different design principles of post-variational strategies combined with linear regression given $\ell_2$ constraints following Theorem~\ref{theorem:constrainedLinear}. The better approach between direct measurement and classical shadows for different design principles is bolded. Given that the Ansatz expansion strategy does not employ multiple observables, the classical shadows method does not provide a speedup, and direct measurement should be used. For the observation construction and hybrid strategy, classical shadows provide a speedup only when the measurement observables are local.}
\label{tab:strategy_bound}
\begin{ruledtabular}
\begin{tabular}{lcc}
$p$ Ans\"atze, $q$ observables, $n$ qubits, $d$ data points&   Direct measurement & Classical shadows\\
\colrule
Ansatz expansion ($q=1$) & $\displaystyle \boldsymbol{\mathcal O\left(\frac{p^2d}{\epsilon^2}\log{\frac{pd}{\delta}}\right)}$ & $\displaystyle \mathcal O\left(\frac{p^2d\|O\|_S^2}{\epsilon^2}\log\frac{pd}{\delta}\right)$\\
Observable construction ($p=1$) & $\displaystyle \mathcal O\left(\frac{q^2d}{\epsilon^2}\log{\frac{qd}{\delta}}\right)$ & $\displaystyle \mathcal O\left(\frac{qd\max_{k\in[q]}\|O\|_S^2}{\epsilon^2}\log{\frac{qd}{\delta}}\right)$\\
Hybrid & $\displaystyle \mathcal O\left(\frac{m^2d}{\epsilon^2}\log{\frac{md}{\delta}}\right)$ & $\displaystyle \mathcal O\left(\frac{mpd\max_{k\in[q]}\|O\|_S^2}{\epsilon^2}\log{\frac{md}{\delta}}\right)$\\
$L$-local Hybrid ($q\in \mathcal O(3^Ln^L)$) & $\displaystyle \mathcal O\left(\frac{9^L n^{2L} Lpd}{\epsilon^2}\log {\frac{npd}{\delta}}\right)$ & $\boldsymbol {\displaystyle \mathcal O\left(\frac{9^L n^{L} Lpd}{\epsilon^2}\log{\frac{npd}{\delta}}\right)}$\\
\end{tabular}
\end{ruledtabular}
\end{table*}

The closed-form solution of linear regression returns the maximum likeihood estimation of the classical parameters $\boldsymbol{\alpha}$, and can be sensitive to noise, resulting in the higher number of shots. To allow for more robustness for error in the estimation, one can employ constrained optimization or a penalty term. 
In the first method, we apply a $\ell_2$ constraint on the parameters such that $\|\boldsymbol{\alpha}\|_2 \le 1$ to the program in Eq.~\ref{eq:alpha} and \ref{eq:hatalpha} and solve with usual convex optimization solvers such as interior point methods~\citep{boyd2004convex}. 
For the second method, we can simply apply Tikhonov regularization penalty with an appropriate ridge parameter $\lambda(\boldsymbol{\alpha})$ to achieve $\|\boldsymbol{\alpha}\|_2 \le 1$. From the Bayesian perspective, we find the maximum a posteriori estimation with a Gaussian prior of variance $(2\lambda(\boldsymbol{\alpha}))^{-1}$ on $\boldsymbol{\alpha}$ to the least squares problem.  For both methods, we can in principle achieve $\|\boldsymbol{\alpha}\|_2 \le 1$, which is an assumption for the next theorem.

\begin{theorem}[Constrained linear regression theoretical guarantee]\label{theorem:constrainedLinear}
Consider a \texttt{(p, q)}-hybrid strategy as in Definition \ref{defHybrid}, and let $m=pq$. Let $Q$ be the matrix as defined in Equation~\ref{eq:q_def}. Let $\epsilon >0$ and let $\hat{Q}$ be such that
\begin{equation*}
\|\hat{Q}-Q\|_{\max} < \frac{\epsilon}{2\sqrt{m}} 
\end{equation*}
Then, for the loss difference as defined in Equation~\ref{eq:rmse_diff}, given the additional constraint that $\|\boldsymbol{\alpha^*}\|_2 \le 1$ and $\|\boldsymbol{\hat \alpha^*}\|_2 \le 1$, has $\Delta\mathcal{L}_{\mathrm{RMSE}} < \epsilon$.
\end{theorem}
This result can be further extended to logistic regression problems with the same $\ell_2$ constraint minimized also on BCE loss given that the sigmoid function is 1-Lipschitz. The proof of the above proposition and extended results can be found in Appendix~\ref{appError}.

For the above results, we see that by applying an $\ell_2$ on $\boldsymbol\alpha$ of the classical regression part of the algorithm, we see that the number of measurements required drops to 
\begin{equation}
t\in\mathcal{O}\left(\frac{m^2d}{\epsilon^2}\log\frac{md}{\delta}\right),
\end{equation}
allowing for more robustness in errors in quantum measurement. A full list of upper bounds for different design principles of post-variational strategies is shown in Table~\ref{tab:strategy_bound}.

\section{Implementation and empirical results}
\label{ch:experiment}
In this section, we discuss the implementation details of our post-variational quantum neural network in the setting of machine learning applications. 

\subsection{Problem setup and implementation details}

To demonstrate applicability on real-world data, we train our post-variational quantum neural network on classification tasks using the Fashion-MNIST dataset~\citep{xiao2017fashion}, which serves as a common benchmark in various quantum machine learning papers~\citep{huang2021power, jerbi2023quantum, alam2021quantum, kerenidis2020quantum, chen2021end}. The Fashion-MNIST dataset consists of 70,000 grayscale images of size $28 \times 28$ of 10 different classes of fashion items and is meant to be a replacement of the MNIST dataset~\citep{lecun1998mnist}. 

As encoding the full image into the quantum circuit may be difficult, we first reduce the dimensions of the image to $4\times 4$ images. As opposed to performing PCA for dimension reduction~\citep{huang2021power, jerbi2023quantum, chen2021end}, we instead apply max pooling over $7\times 7$ patches and rescaling the parameters to a range of $[0, 2\pi)$. We then encode each column into a single qubit by iterating between $RZ$ and $RX$ gates similar to \citep{alam2021quantum} as shown in Figure~\ref{fig:data_encode}. The reasoning of using this particular method is that Fashion-MNIST is still a relatively easy dataset and the usage of PCA may perform enough substantial heavylifting such that the task becomes less non-trivial.

For the Ansatz used in the variational baseline and the backbone for the Ansatz expansion and hybrid strategies for our post-variational methods, we use a simple Ansatz made of 2 alternations of $RY$ gates and circular $CNOT$ gates as shown in Figure~\ref{fig:ansatz}. We set initial parameters to 0, on which the Ansatz would evaluate to identity. Such Ansatz designs and parameter initialization have been 
shown to be able to avoid barren plateaus in variational algorithms~\citep{grant2019initialization}.

To train our post-variational quantum neural networks, we generate post-variational feature embeddings by Algorithm~\ref{algo:main} via \texttt{qiskit}~\citep{research2017qiskit}. For the classical regression layer, we use the logistic regression algorithm as provided by the \texttt{scikit-learn} library~\citep{pedregosa2011scikit}.

\begin{figure}
\centering
\includegraphics[width=0.95\linewidth]{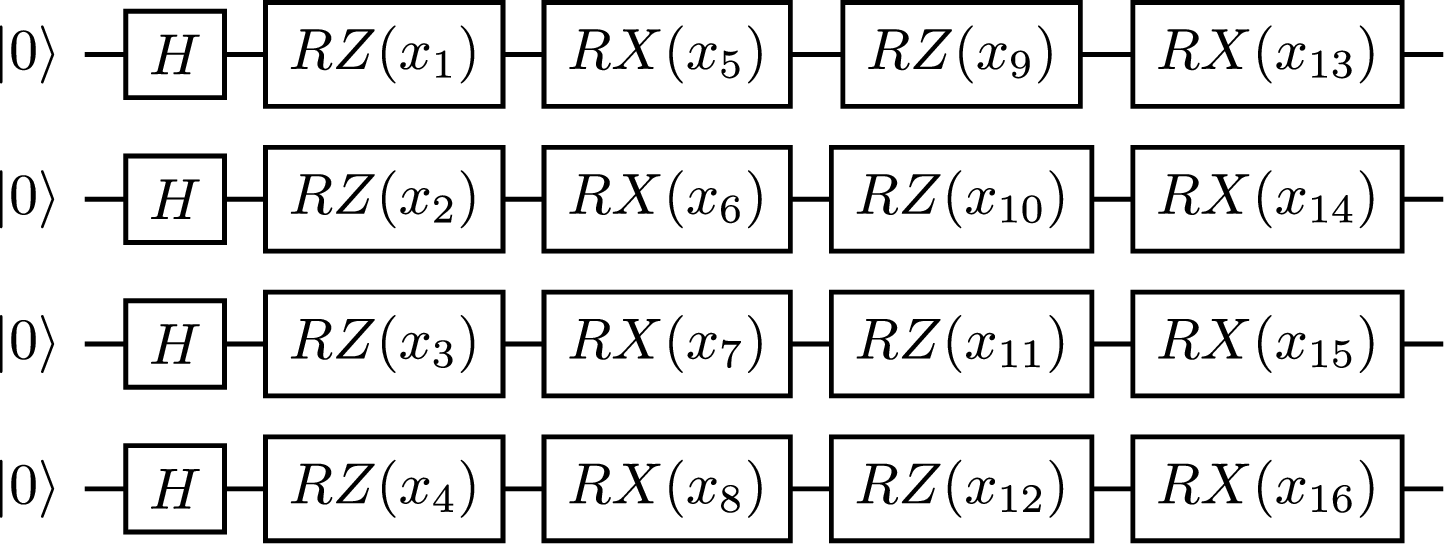}
\caption{Data encoding circuit. Each column of the compressed image is encoded into a single qubit, and each row is encoded consecutively via alternating rotation-Z and rotation-X gates.}
\label{fig:data_encode}
\end{figure}

\begin{figure}
\centering
\includegraphics[width=0.95\linewidth]{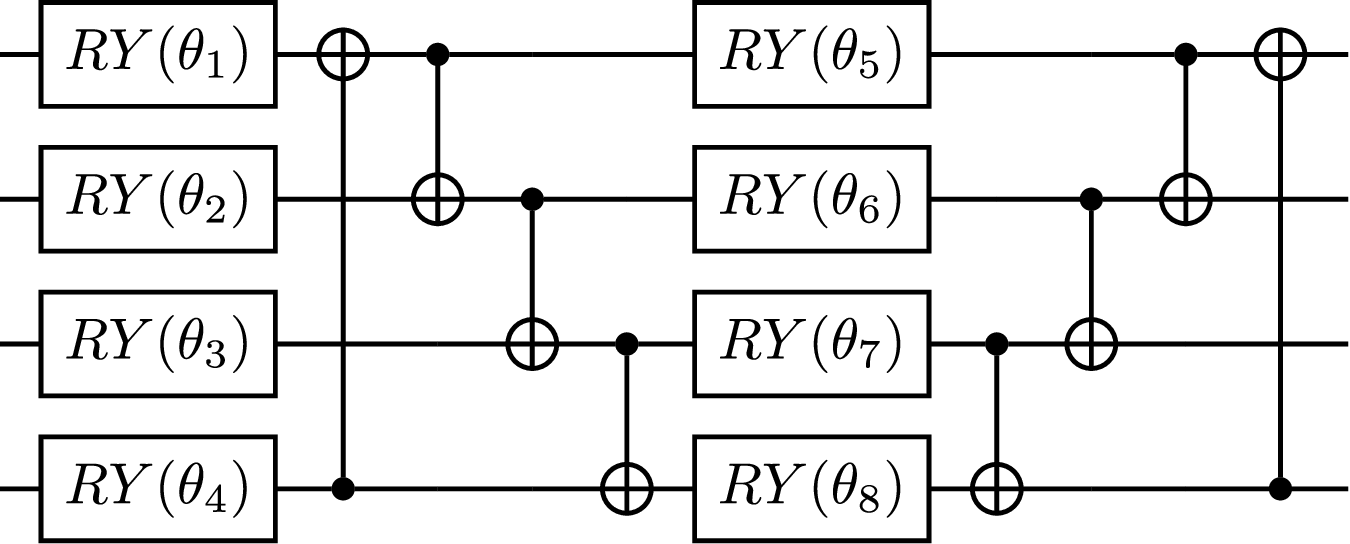}
\caption{Ansatz used in variational algorithm. This Ansatz is also used as the Ans\"atze generating backbone for the Ansatz expansion and hybrid post-variational strategies. We set all initial parameters to 0, by which the Ansatz evaluates to identity.}
\label{fig:ansatz}
\end{figure}

\begin{algorithm}
\caption{Hybrid PV feature generation}
\label{algo:main}
\DontPrintSemicolon

\KwIn{qubit number $n\in\mathbb{N}$, locality $L \in\mathbb{N}$,
derivative order $R\in\mathbb{N}$, dataset $\{x_i\}_{i\in[d]}$, data encoding unitary $S(\cdot)$, Ansatz $U(\cdot )$}
\KwOut{Post-variational embedding $Q \in \mathcal M_{m\times d}(\mathbb R)$}

$shifts \gets \text{ListOfShiftedParamsInDerivativeOrder}(R)$\;
$observables \gets \text{ListOfPaulisInLocality}(L)$\;
\For {$i\in[d]$} {
    \For {$(j, \boldsymbol{\theta}, O) \in \texttt{enum}(shifts \times observables)$}{
        $Q_{ij} =\langle 0^n |S^\dagger(\boldsymbol{x_i}) U^\dagger(\boldsymbol{\theta})OU(\boldsymbol{\theta}) S(\boldsymbol{x_i})| 0^n\rangle$\;
    }
}
\Return $Q$
\end{algorithm}
\subsection{Experimental results}

We experiment on the effectiveness of post-variational quantum neural networks on a simplified version of the FashionMNIST classification task, where we conduct binary classification of the classes \texttt{coat} and \texttt{shirt}, training on 200 samples and testing on 50 samples from each class. Experimental results are shown in Table~\ref{tab:main}. We do not include the loss of variational algorithms as they use the variational Hamiltonian loss function as opposed to the binary cross-entropy loss of classical and post-variational quantum neural networks.

Our results show that all post-variational methods exceed the variational algorithm while using the same Ansatz for the Ansatz expansion and hybrid strategies. From comparisons with the classical linear estimator based logistic regression models, we note that including 2-local observables or 1-local observables with 2-order derivatives, we can exceed the performance of classical linear models. With 3-local observables or 2-local observables with 1-order derivatives, the post-variational neural networks can even exceed the performance of a two-layer multilayer perceptron in train accuracy. However, given that the post-variational algorithms extracts more features than the classical algorithm, there are more parameters to optimize, leading to overfitting on the training model to a certain extent, as shown by the decreasing testing accuracy of these models.

From the performance results of these results, we can obtain a glimpse of the effectiveness of each strategy. While the Ansatz expansion strategy does not perform much better even when we take up to second-order derivatives, the inclusion of second-order derivatives provide a large boost in accuracy when we use derivatives in conjunction with all 1-local observables in the hybrid strategy. Further, the addition of first-order derivatives also provide a boost in accuracy when added to 2-local observables. This implies that the addition of gradient circuits should serve mostly as an heuristic to expand the expressibility to observable construction methods and may not be sufficient in itself as a good training strategy, as it is limited by the expressibility of the structure of the Ansatz.

\begin{table}
\caption{Effectiveness of different design principles of post-variational quantum neural networks. We note that post-variational methods can outperform variational methods, and can even exceed small scale multilayer perceptrons in training accuracy.}
\label{tab:main}
\begin{ruledtabular}
\begin{tabular}{llcccc}
            &                   & \multicolumn{2}{c}{Train} & \multicolumn{2}{c}{Test} \\
            &                   & Loss        & Acc.        & Loss       & Acc.       \\
\colrule
Classical   & Logistic          & 0.5379      & 69.25\%     & 0.5913     & 65.33\%     \\
            & MLP               & 0.4457      & 77.92\%     & 0.7176     & 67.67\%     \\
\colrule
Variational &                   & -           & 55.83\%     & -          & 50.67\%     \\
\colrule
Ansatz      & 1-order           & 0.6849      & 56.08\%     & 0.6996     & 55.00\%     \\
            & 2-order           & 0.6593      & 57.75\%     & 0.7078     & 53.67\%     \\
\colrule
Observable  & 1-local           & 0.6228      & 65.42\%     & 0.6630     & 60.00\%     \\
            & 2-local           & 0.5441      & 72.42\%     & 0.7313     & 58.67\%     \\
            & 3-local           & 0.4610      & 78.67\%     & 0.7482     & 59.67\%     \\
\colrule
Hybrid      & 1-order + 1-local & 0.5912      & 67.33\%     & 0.6977     & 61.67\%     \\
            & 2-order + 1-local & 0.4971      & 75.42\%     & 0.8017     & 55.67\%     \\
            & 1-order + 2-local & 0.4337      & 78.00\%     & 0.8881     & 57.67\%    
     
\end{tabular}
\end{ruledtabular}
\end{table}

\begin{table}
\caption{Training results of multiclass classification.}
\label{tab:multi}
\begin{ruledtabular}
\begin{tabular}{lcccc}
         & Logistic & MLP    & Variational & 1-order+2-local PV \\
\colrule
Loss     & 0.8246   & 0.4865 & -           & 0.6786                    \\
Accuracy & 0.6725   & 0.815  & 0.1675      & 0.825                    
\end{tabular}
\end{ruledtabular}
\end{table}

Finally, we mention the capacities of such models in the training of multiclass data. Multiclass training in variational algorithms involves a variety of different strategies, including partitioning the output of measured bitstrings~\citep{bokhan2022multiclass}, conducting measurements on different qubits~\citep{chen2021end}, training different variational algorithms in an one-versus-all setting~\citep{schuld2020circuit}, or attaching a classical neural network to the measurement outcome of the qubits of the variational algorithms~\citep{alam2021quantum}. Post-variational quantum neural networks, on the other hand, can be extended to multiclass problems, being simply adding an additional dimension to the classical linear map, as it supplies a larger number of feature extractors, and leans into the power of finding different combinations of quantum feature maps, making them a natural candidate for multiclass tasks. Our experimental results on training 400 evenly sampled classes for multiclass classification show that post-variational quantum neural network can obtain comparable performance to that of multi-layer perceptrons in such types of training as well (Table~\ref{tab:multi}).

\section{Discussion}

In this work, we introduce post-variational quantum neural networks as an alternative and non-iterative method for quantum machine learning in the NISQ era~\citep{preskill2018quantum}. In contrast to variational algorithms, post-variational quantum neural networks conduct measurements of fixed quantum circuits and find the optimal classical combination of such circuits to form a machine learning model, rather than directly optimizing over variational parameters on quantum computers.

Given a well-selected set of good fixed Ans\"atze, the post-variational method involves only convex optimization, and hence can provide a guarantee of finding a global minimum solution in polynomial time in regards to the selected set of Ans\"atze. We emphasize again that while this property of post-variational methods provides a terminable algorithm and optimal solution conditioned on the set of Ans\"atze given, we do not claim to resolve barren plateau problems or related exponential concentration stemming from the exponentially large Hilbert space. The hardness of the problem is instead delegated to the selection of the set of fixed Ans\"atze from an exponential amount of possible quantum circuits, which we attempt to mitigate using our three heuristical strategies -- Ansatz expansion, observable construction, and a hybrid method of the former two. The effectiveness of these methods can then be seen empirically in our experimental results in Section~\ref{ch:experiment}.

While we have an algorithm that terminates, we do not exhibit a quantum advantage that is proven. We retain the data embedding and the fixed part of the parameterized circuits. Hence, we still retain the classical hardness of simulating these circuits~\citep{bremner2016average,bremner2017achieving,havlivcek2019supervised} and the potential, if any, of quantum advantages for certain data sets and problems.

The implementation and limitations of neural networks on quantum computers has been extensively studied in both fault-tolerant~\citep{altaisky2001quantum, schuld2014quest, wiebe2015quantum, schuld2015simulating, kapoor2016quantum, amin2018quantum, rebentrost2018quantum, allcock2020quantum, kerenidis2020quantum, guo2024quantum} and variational settings on NISQ hardware~\citep{cerezo2021variational, mitarai2018quantum, schuld2020circuit, farhi2018classification, benedetti2019parameterized, schuld2019quantum, cong2019quantum}. 
In comparison to these prior works, we now provide some arguments that post-variational quantum neural networks may have fewer quantum gates per circuit, as well as potentially lower circuit depth.
For fault-tolerant implementations of neural networks on quantum hardware, a construction of classical neural networks with potential quantum speedups can be achieved by techniques such as Grover's algorithm~\citep{grover1996fast}, amplitude amplification~\citep{brassard2002quantum}, quantum signal processing~\citep{low2017optimal, low2019hamiltonian}, quantum singular value transformation~\citep{gilyen2019quantum}, and/or nonlinear amplitude transformations~\citep{guo2021nonlinear,rattew2023nonlinear}. While with these techniques one is able to argue provable advantages, see, e.g., \citep{allcock2020quantum}, the setting of these algorithms requires deeper and more complex quantum circuits than that of quantum algorithms suitable for the NISQ era. 

On the other hand, while comparing to variational algorithms, we note that by using our heuristic strategies, we can also potentially lower the number of quantum gates per quantum circuit. By replacing part of the Ansatz with an ensemble of local Pauli measurements as with our observable construction method, one reduces the depth of the circuit. Using the Ansatz expansion strategy results in fixed circuits. These fixed circuits we can optimize with transpilation and circuit optimization strategies. Often the our initial circuit has the parameters to set to zero, and we can remove gates that evaluate to identity. Obvious examples would be single-qubit rotational gates or two CNOT gates that cancel each other out, leading to fewer gates per circuit, and potentially lower circuit depth.

While our empirical results show that there are cases where the usage of post-variational quantum neural networks surpass the performance of variational algorithm, we do not make a statement on the superiority of variational and post-variational algorithms as different problem settings may lead to different algorithms outperforming the other. We propose post-variational quantum neural networks simply as an alternative implementation of neural networks in the NISQ setting, and leave the determination of case-by-case distinctions on performance evaluations and resource consumption to future work.

\begin{acknowledgments}
The authors would like to thank Naixu Guo, Xiufan Li and Beng Yee Gan for valuable discussions. The authors would also like to thank Rahul Jain and Warut Suksompong for the constructive criticism of an early draft of this manuscript, and Ying Yang for help proofreading the original preprint. This research is supported by the National Research Foundation, Singapore, and A*STAR under its CQT Bridging Grant and its Quantum Engineering Programme under grant NRF2021-QEP2-02-P05.
\end{acknowledgments}

\onecolumngrid
\appendix

\section{Deconstruction of a variational Ansatz}
\label{appendix:decompose}
Here we provide a simple argument to deconstruct a variational Ansatz and show that given enough terms, the post-variational algorithm can retrieve the optimum obtained by the variational Ansatz with a terminable guarantee.

As quantum circuits are specified by a sequential arrangement of quantum logic gates, the variational Ansatz $U(\boldsymbol\theta)$ of a variational quantum circuit that is parameterized by $\boldsymbol{\theta} \in \mathbb{R}^s$ can be expressed as a product of unitaries such that 
\begin{equation}
U(\boldsymbol\theta) = U_1(\theta_1)U_2(\theta_2)\cdots U_s(\theta_s).
\end{equation}
For the sake of simplicity, we discuss only one-parameter quantum logic gates for our analysis of the variational algorithm. 

\begin{theorem}
[Stone's theorem on one-parameter unitary groups~\citep{stone1930linear, stone1932one}] A one-to-one correspondence between Hermitian operators $H$ on a Hilbert space $\mathcal{H} = {\mathbb{C}^{2^n}}$ and one-parameter families $(U_t)_{t\in\mathbb{R}}$ of strongly continuous homomorphic unitaries can be shown as follows:
$$U_t = e^{itH}.$$
\end{theorem}
Formally, the variational Ansatz $U(\boldsymbol\theta)$ can be expanded using Stone's theorem as,
\begin{equation}
U(\boldsymbol\theta) = W_1e^{i\theta_1H_1}V_1W_2e^{i\theta_2H_2}V_2 \cdots W_se^{i\theta_sH_s}V_s.
\end{equation}
where $W_i$ and $V_i$ are fixed unitaries and $H_i$ are Hermitian operators.

\begin{theorem}
[Baker–Campbell–Hausdorff identity~\citep{campbell1897law}]
Given two matrices $X$ and $Y$, $$e^XYe^{-X} = \sum_{n=0}^\infty\frac{[(X)^n,Y]}{n!},$$
where $[(X)^{n},Y]\equiv[\underbrace{X,\cdot\cdot\cdot[X,[X}_{n{\mathrm{~times}}},Y]]$ and $[(X)^{0},Y]\equiv Y.$
\end{theorem}
We apply this identity recursively to rewrite the Hermitian operator $U^\dagger(\theta) OU(\theta)$ in polynomial terms of $\theta_1, \theta_2, \cdots \theta_s$. Expanding the first unitary $U_1$, we get

\begin{equation}
\label{unitary_expansion}
U_1^\dagger(\theta_1) O U_1(\theta_1) = V_1^\dagger\left(\sum_{k=0}^\infty\frac{[(i\theta_1H_1)^k, W_1^\dagger O W_1]}{k!}\right) V_1 = \sum_{k=0}^\infty\frac{\theta_1^k}{k!}V_1^\dagger[(iH_1)^k, W_1^\dagger O W_1]V_1.
\end{equation}

We note that the term $iH_j$ is anti-Hermitian, hence for all $k \ge 0$, $V_1^\dagger[(iH_1)^k,  W_1^\dagger O W_1]V_1$ is Hermitian. Thus, we can write $U_1^\dagger(\theta_1) O U_1(\theta_1)$ as a weighted polynomial sum of Hermitians against $\theta_1$. Plugging the result recursively against the other terms, one can obtain  
\begin{multline}
\label{unitary_full_expansion}
U^\dagger(\boldsymbol\theta) O U(\boldsymbol\theta) = \sum_{\forall i \in [s], a_i \ge 0}\left(\prod_{i=1}^s\frac{\theta_i^{a_i}}{a_i!}\right)
\underbrace{V_s^\dagger[(iH_s)^{a_s}, W_s^\dagger V_{s-1}^\dagger[(iH_{s-1})^{a_{s-1}}, \cdots V_1^\dagger[(iH_1)^{a_1}, W_1^\dagger O W_1]V_1\cdots ]V_{s-1}W_s]V_s}_{\mathrm{Hermitian}}.
\end{multline}
This expansion allows one to express 
\begin{equation}
 \mathfrak{O}(\boldsymbol{\theta}) = U(\boldsymbol{\theta})OU^\dagger(\boldsymbol{\theta}) \approx \sum_i \mathcal{F}_i(\boldsymbol{\theta}) \mathfrak{O}_i.
\end{equation} 
This discussion shows that a deconstruction of the variational algorithm does indeed exist formally and corresponds to a linear combination of Hermitians. However, we note that there are currently infinite terms in this expression. We now show that we can express this linear combination in limited terms.
Any Hermitian operator can be expressed in a basis of Pauli matrices such that 
\begin{equation}
H\in {\mathcal M_{2\times 2}(\mathbb{C})}^{\otimes n} \rightarrow H \in \operatorname{span} \left(\{\mathcal{I}, \mathcal{X}, \mathcal{Y}, \mathcal{Z}\}^{\otimes n}\right),
\end{equation}
hence, we can state that 
\begin{equation}
U^\dagger(\boldsymbol{\theta}) O U(\boldsymbol{\theta})\in \operatorname{span}\left(\{\mathcal{I}, \mathcal{X}, \mathcal{Y}, \mathcal{Z}\}^{\otimes n}\right),
\end{equation}
where each coefficient is $\poly(\theta_1, \theta_2, \cdots \theta_s)$.

Therefore, supposing that the variational algorithm obtains an optimal $\boldsymbol{\theta^*}$, the post-variational algorithm requires at most $4^n$ terms to express the same optimal answer. Comparing the number of parameters tunable, we see that the variational algorithm requires $\mathcal O(\poly(s))$ parameters, while the post-variational algorithm requires $\mathcal O(4^n)$ parameters. We note that the quantum advantage of variational algorithms in terms of parameters is the ability to generate different observables on higher orders of $\boldsymbol{\theta}$, a feat that activation functions of classical computers cannot achieve, as classical activation functions can only non-linearize the $\boldsymbol{\theta}$ parameter itself without generating new observables. 
However, considering that the variational algorithm returns an estimation of the optimal answer, rather than the exact optimal answer, we hope to restrict the number of Hermitian terms used in the post-variational algorithm to $\mathcal O(\poly(s))$ terms to achieve a similar estimation. 

\section{Derivation of random measurement errors}
\label{appMeasure}

\begin{proof}[Proof of Proposition \ref{proposition:measure}]
We represent the output of a single quantum neuron over a single data output over multiple runs as a series of i.i.d. random variables $Z_1, Z_2, \cdots Z_t$, each with a range of $[-1, 1]$. Then by Hoeffing bound~\citep{hoeffding1963probability}, we find that the sample mean $\bar{Z}$ and the expectation of the above random variables have 
$\Pr(|\bar{Z} - E[Z]| \ge \epsilon) \le 2\exp(-\frac{t\epsilon_H^2}{2}).$
Further, as each run is associated with a single data point on a single quantum neuron, we need to apply a union bound~\citep{boole2009mathematical} to ensure for all $m$ quantum neurons, and $d$ data points, we only have a small probability of $|\bar{Z} - E[Z]| \ge \epsilon_H$. We denote the event $E_{i, j}$ as the observation $i$th data point and $j$th quantum neuron having $|\bar{Z} - E[Z]| \ge \epsilon_H$ for the corresponding observations. We then see that
$
\Pr(\cup_{i=1}^m \cup_{j=1}^d E_{ij}) \le\sum_{i=1}^m\sum_{j=1}^d\Pr(E_{ij}) = md\cdot \Pr(|\bar{Z} - E[Z]| \ge \epsilon_H) \le 2md\exp(-\frac{t\epsilon_H^2}{2})\le \delta, 
$
such that $t \in \mathcal O(\frac{1}{\epsilon_H^{2}}\log \frac{md}{\delta})$. We note the number of quantum measurements are duplicated across $m$ quantum neurons and $d$ data points, hence the total number of observations required is 
$\mathcal O(\frac{md}{\epsilon_H^2}\log \frac{md}{\delta}).$
\end{proof}

\begin{proof}[Proof of Proposition \ref{proposition:shadow}]
For simplicity, we consider the original classical shadows protocol~\citep{huang2020predicting} utilizing the median-of-means estimator~\citep{nemirovsky1983problem} for discussion. With the median-of-means estimator, one partitions the number of measurements into $s$ groups of $t$ measurements, taking the median of the mean over the $t$ measurements from the $s$ groups.
From the analysis of the median-of-means estimator in \citep{huang2020predicting}, we note that with probability $1-\delta$, per Hoeffding bound~\citep{hoeffding1963probability} and Chebyshev's inequality, one should set $t \in \mathcal O\left(\frac{\max_{k \in [q]}\|O_k\|_S^2}{\epsilon_{H}^2}\right)$ and per union bound, one should set $s \in \mathcal O(\log\frac{md}{\delta})$. As the classical shadows method predicts the outcome of $q$ quantum neurons that share the same Ansatz and data point, the process has to be duplicated over $d$ data points and $p$ different Ans\"atze. Hence the total number of measurements needed in total is then $\mathcal O\left(\frac{pd\max_{k \in [q]}\|O_k\|_S^2}{\epsilon_H^2}\log \frac{md}{\delta}\right).$
\end{proof}

\section{Derivation of error propagation}
\label{appError}
In this section, we detail the derivation of the error propagation of the element-wise error of the quantum neuron to the entire network. Before the proofs, we first set up some further notations needed.
Given a field $\mathbb{F}$ of either real or complex numbers, for matrix $A\in \mathcal M_{M\times N}(\mathbb{F})$, we denote the Frobenius norm, or Hilbert-Schmidt norm, to be $\|A\|_F = \sqrt{\sum_i\sum_j |A_{ij}|^2} = \sqrt{\sum_{i=1}^{\min(M, N)} \sigma_i^2(A)}$. We also denote the nuclear norm to be $\|A\|_* = \sum_{i=1}^{\min(M, N)} \sigma_i(A)$. Recall that $\sigma_i(A)$ is the $i$-th singular value of $A$.

Before we prove Theorem \ref{theorem:linear}, we first show some preliminaries to aid the proof of our results. For matrix $A \in \mathcal{M}_{M\times N}(\mathbb{R})$, we note the following inequalities for matrix norms
\begin{align}
    \|A\|_{\max} \le \|A\| &\le \|A\|_F \le \sqrt{MN}\|A\|_{\max},\\
    \|A\|_F \le \|A\|_* &\le \sqrt{r}\|A\|_F \le r \|A\|,
\end{align}
where $r = \rank(A) \le \min(M,N)$. 

\begin{theorem} [Perturbation inequality for concave functions of singular values~\citep{oymak2011simplified, yue2016perturbation}]
\label{theoremPerturbationConcave}
    Given matrices $A, B \in \mathcal{M}_{M\times N}(\mathbb{R})$. Suppose that $f:\mathbb{R}_+ \to \mathbb{R}_+$ is a concave function where $f(0) = 0$, then we have
    $$\sum_{i=1}^{\min(M, N)}|f(\sigma_i(A))-f(\sigma_i(B))| \le \sum_{i=1}^{\min(M, N)} f(\sigma_i(A-B)).$$
\end{theorem}
We use this theorem to prove the follow lemma.
\begin{lemma} [Perturbation inequality for rank differences]
\label{lemmaPerturbationRank}
    Given matrix $A, B \in \mathcal{M}_{M\times N}(\mathbb{R})$, if $\|A-B\|_{\max} < \min(\sigma_{\min}(A), \sigma_{\min}(B)) / \sqrt{\min(M,N)MN}$, where $\sigma_{\min}(A)$ is the smallest non-zero singular value of $A$, then $\rank(A) = \rank(B)$.
\end{lemma}
\begin{proof}
    We define concave function $f_\mu: \mathbb{R}_+ \to [0, 1]$ such that 
    \begin{equation}
         f_\mu(x) := \begin{cases} \frac{x}{\mu}, & 0 \le x < \mu,\\ 1, & \mu \le x.\end{cases}
    \end{equation}
    Suppose that $\mu =\min(\sigma_{\min}(A), \sigma_{\min}(B))$, we find that  
    \begin{equation}
        \sum_{i=1}^{\min(M, N)} \lvert f_\mu(\sigma_i(A)) - f_{\mu}(\sigma_i(B))\rvert = \lvert\rank(A)-\rank(B)\rvert.
    \end{equation}
    Then by Theorem~\ref{theoremPerturbationConcave}, we see that 
    \begin{equation}
        \lvert\rank(A)-\rank(B)\rvert \le \sum_{i=1}^{\min(M, N)} f_\mu(\sigma_i(A-B)) \le \sum_{i=1}^{\min(M, N)} \frac{\sigma_i(A-B)}{\mu} = \frac{\|A-B\|_*}{\min(\sigma_{\min}(A), \sigma_{\min}(B))}
    \end{equation}
    Using the matrix norm inequalities, we see that
    \begin{multline}
        |\rank(A)-\rank(B)| \le \frac{\|A-B\|_*}{\min(\sigma_{\min}(A), \sigma_{\min}(B))} \le \frac{\sqrt{\min(M,N)}\|A-B\|_F}{\min(\sigma_{\min}(A), \sigma_{\min}(B))} \le \frac{\sqrt{\min(M,N)MN}\|A-B\|_{\max}}{\min(\sigma_{\min}(A), \sigma_{\min}(B))}
    \end{multline}
    Setting an upper bound of 1 for the above inequality, we find that if
    $\frac{\sqrt{\min(M,N)MN}\|A-B\|_{\max}}{\min(\sigma_{\min}(A), \sigma_{\min}(B))} < 1$, or 
    \begin{equation}
        \|A-B\|_{\max} < \frac{\min(\sigma_{\min}(A), \sigma_{\min}(B))}{\sqrt{\min(M,N)MN}},
    \end{equation}
    then $|\rank(A)-\rank(B)| < 1$. Given that ranks have positive integers values, we find that $\rank(A) = \rank(B)$.
\end{proof}
We also use the following theorem:
\begin{theorem}[Perturbation theory for pseudoinverses~\citep{wedin1973perturbation}]
\label{theorem:perturbation_norm}
The spectral norm of the difference of the pseudoinverses of two matrices $A, B\in \mathcal{M}_{m\times n}(\mathbb{R})$ can be upper bounded as follows if $\rank(A) = \rank(B)$:
$$\|B^+-A^+\|\le 2 \|A^+\|\|B^+\|\|B-A\|.$$
\end{theorem}
We continue with putting these technical results together to prove the Theorem \ref{theorem:linear} of the main text.
\begin{proof} [Proof of Theorem \ref{theorem:linear}]
Recall that 
\begin{equation}
\Delta \mathcal{L}_{\mathrm{RMSE}} = \mathcal{L}_{\mathrm{RMSE}}(\boldsymbol{\hat{\alpha}^*}, Q) - \mathcal{L}_{\mathrm{RMSE}}(\boldsymbol{\alpha^*}, Q).
\end{equation}
Expanding directly, we obtain 
\begin{align}
\Delta \mathcal{L}_{\mathrm{RMSE}} &\le \lvert\Delta \mathcal{L}_{\mathrm{RMSE}}\rvert \le \frac{1}{\sqrt{d}}\left\lvert\|\boldsymbol{Y}-Q\boldsymbol{\hat{\alpha}^*}\|_2 - \|\boldsymbol{Y}-Q\boldsymbol{\alpha^*}\|_2\right\rvert
\\
\text{(reverse triangle)} &\le \frac{1}{\sqrt{d}}\|Q(\boldsymbol{\hat{\alpha}^*}-\boldsymbol{\alpha^*})\|_2 \le \frac{1}{\sqrt{d}}\|Q\| \|\boldsymbol{\hat{\alpha}^*}-\boldsymbol{\alpha^*}\|_2 \label{eq:rmse_bound}\\ 
&= \frac{1}{\sqrt{d}}\|Q\| \|\hat Q^+\boldsymbol{Y}-Q^+\boldsymbol{Y}\|_2 \le \frac{1}{\sqrt{d}}\|Q\|\|\boldsymbol{Y}\|_2\|\hat Q^+-Q^+\| .
\end{align}

With $0 < \epsilon_0 < 1$, we conduct enough measurements such that
$\|\hat{Q}-Q\|_{\max} < \epsilon_0 \frac{\min(\sigma_{\min}(Q), \sigma_{\min}(\hat{Q}))}{\sqrt{\min(m, d) md}}$, hence by Lemma~\ref{lemmaPerturbationRank}, $\rank(Q) = \rank(\hat{Q})$. Then by Theorem~\ref{theorem:perturbation_norm}, we obtain
\begin{align}
\|\hat{Q}^+ - Q^+\| &\le 2\|Q^+\|\|\hat{Q}^+\|\|\hat{Q}-Q\| \le 2\|Q^+\|\left(\|Q^+\| + \|\hat{Q}^+ - Q^+\|\right)\|\hat{Q}-Q\|\\
\Rightarrow \|\hat{Q}^+ - Q^+\| &\le \frac{2\|Q^+\|^2\|\hat{Q}-Q\|}{1-2\|Q^+\|\|\hat{Q}-Q\|}.
\end{align}
We note that if $m,d\ge 9$, then 
\begin{equation}
\|Q^+\|\|\hat{Q}-Q\| \le \frac{\sqrt{md}}{\sigma_{\min}(Q)}\|\hat{Q}-Q\|_{\max} < \frac{\sqrt{md}}{\sigma_{\min}(Q)} \frac{\min(\sigma_{\min}(Q), \sigma_{\min}(\hat{Q}))}{\sqrt{\min(m, d) md}} \le \frac{1}{\sqrt{\min(m, d)}} \le \frac{1}{3}.
\end{equation}
Hence we see that 
\begin{equation}
\|\hat{Q}^+ - Q^+\| \le 6 \|Q^+\|^2\|\hat{Q}-Q\|.
\end{equation}
Evaluating the loss
\begin{align}
\Delta\mathcal{L}_{\mathrm{RMSE}} &< \frac{6} {\sqrt{d}}\|\boldsymbol{Y}\|_2 \|Q\|\|Q^+\|^2 \|\hat{Q}-Q\| \leq 
\frac{6 \sqrt{md}} {\sqrt{d}}\|\boldsymbol{Y}\|_2 \|Q\|\|Q^+\|^2 \|\hat{Q}-Q\|_{\max} \\
&< 6 \|\boldsymbol{Y}\|_2 \|Q\|\|Q^+\|^2  \frac{\min(\sigma_{\min}(Q), \sigma_{\min}(\hat{Q}))}{\sqrt{\min(m, d) d}} \epsilon_0.
\end{align}
Now setting (it has to satisfy the upper bound $1$)
\begin{equation}
\epsilon_0 \leq \min \left(1,  \frac{\epsilon \sqrt{\min(m, d) d}}{6 \|\boldsymbol{Y}\|_2 \|Q\|\|Q^+\|^2 \min(\sigma_{\min}(Q), \sigma_{\min}(\hat{Q}))} \right),
\end{equation}
 then $\Delta \mathcal{L}_{\mathrm{RMSE}}\le \epsilon$.
We hence obtain for element-wise bound
\begin{equation}
\|\hat{Q}-Q\|_{\max} < \epsilon_0 \frac{\min(\sigma_{\min}(Q), \sigma_{\min}(\hat{Q}))}{\sqrt{\min(m, d) md}} \leq \min \left(\frac{\min(\sigma_{\min}(Q), \sigma_{\min}(\hat{Q}))}{\sqrt{\min(m, d) md}},  \frac{\epsilon }{6 \sqrt{m}\|\boldsymbol{Y}\|_2 \|Q\|\|Q^+\|^2 } \right). 
\end{equation}
\end{proof}

Finally, we show the case of $\ell_2$-constrained linear regression and logistic regression.
\begin{proof} [Proof of Theorem \ref{theorem:constrainedLinear}]
Note that due to the definition of $\boldsymbol{\hat\alpha^*}$, 
\begin{equation}
\|\boldsymbol{Y}-\hat Q\boldsymbol{\hat\alpha^*}\|_2 \le \|\boldsymbol{Y}-\hat Q\boldsymbol{\alpha^*}\|_2
\end{equation}
Expanding $\Delta \mathcal{L}_{\mathrm{RMSE}}$ and using the above property, we obtain 
\begin{align}
\Delta \mathcal{L}_{\mathrm{RMSE}} &=
\mathcal{L}_{\mathrm{RMSE}}(\boldsymbol{\hat{\alpha}^*}, Q) - \mathcal{L}_{\mathrm{RMSE}}(\boldsymbol{\alpha^*}, Q)
\\
&= \mathcal{L}_{\mathrm{RMSE}}(\boldsymbol{\hat{\alpha}^*}, Q) - \mathcal{L}_{\mathrm{RMSE}}(\boldsymbol{\hat\alpha^*}, \hat Q) + \mathcal{L}_{\mathrm{RMSE}}(\boldsymbol{\hat\alpha^*}, \hat Q) -\mathcal{L}_{\mathrm{RMSE}}(\boldsymbol{\alpha^*}, Q)\\
&\le \mathcal{L}_{\mathrm{RMSE}}(\boldsymbol{\hat{\alpha}^*}, Q) - \mathcal{L}_{\mathrm{RMSE}}(\boldsymbol{\hat\alpha^*}, \hat Q) + \mathcal{L}_{\mathrm{RMSE}}(\boldsymbol{\alpha^*}, \hat Q) -\mathcal{L}_{\mathrm{RMSE}}(\boldsymbol{\alpha^*}, Q)\\
&\le |\mathcal{L}_{\mathrm{RMSE}}(\boldsymbol{\hat{\alpha}^*}, Q) - \mathcal{L}_{\mathrm{RMSE}}(\boldsymbol{\hat\alpha^*}, \hat Q)| + |\mathcal{L}_{\mathrm{RMSE}}(\boldsymbol{\alpha^*}, \hat Q) -\mathcal{L}_{\mathrm{RMSE}}(\boldsymbol{\alpha^*}, Q)|\\
&= \frac{1}{\sqrt{d}}\left(\left|\|\boldsymbol{Y}-Q\boldsymbol{\hat{\alpha}^*}\|_2 - \|\boldsymbol{Y}-\hat Q\boldsymbol{\hat\alpha^*}\|_2\right|+\left|\|\boldsymbol{Y}-\hat Q\boldsymbol{\alpha^*}\|_2 - \|\boldsymbol{Y}-Q\boldsymbol{\alpha^*}\|_2\right|\right)\\
\text{(reverse triangle)}&\le \frac{1}{\sqrt{d}}\left(\|Q\boldsymbol{\hat{\alpha}^*}-\hat Q\boldsymbol{\hat\alpha^*}\|_2+\|\hat Q\boldsymbol{\alpha}^*-Q\boldsymbol{\alpha^*}\|_2\right)\\
&\le \frac{2}{\sqrt{d}}\|Q-\hat Q\| \le 2\sqrt{m} \|Q-\hat Q\|_{\max}
\end{align}
Setting $\Delta \mathcal{L}_{\mathrm{RMSE}}\le \epsilon$, we obtain for element-wise bound
\begin{equation}
\|\hat{Q}-Q\|_{\max} < \frac{\epsilon}{2\sqrt{m}} 
\end{equation}
\end{proof}

To extend the above results to logistic regression, simply change the estimator from $Q\boldsymbol{\alpha}$ to $\sigma(Q\boldsymbol{\alpha})$ and consider binary cross entropy loss. Similarly defining
\begin{align}
\Delta \mathcal{L}_{\mathrm{BCE}} &= \mathcal{L}_{\mathrm{BCE}}(\boldsymbol{\hat{\alpha}^*}, Q) -\mathcal{L}_{\mathrm{BCE}}(\boldsymbol{\alpha^*}, Q)\\
&\le |\mathcal{L}_{\mathrm{BCE}}(\boldsymbol{\hat{\alpha}^*}, Q) - \mathcal{L}_{\mathrm{BCE}}(\boldsymbol{\hat\alpha^*}, \hat Q)| + |\mathcal{L}_{\mathrm{BCE}}(\boldsymbol{\alpha^*}, \hat Q) -\mathcal{L}_{\mathrm{BCE}}(\boldsymbol{\alpha^*}, Q)|\\
&\begin{multlined}[b]
    \le\frac{1}{d}\left|\boldsymbol{Y}^\intercal \left(\log(\sigma(Q\boldsymbol{\hat\alpha^*})) - \log(\sigma(\hat Q\boldsymbol{\hat\alpha^*}))\right)+ (\boldsymbol{1}-\boldsymbol{Y})^\intercal \left(\log(1-\sigma(Q\boldsymbol{\hat\alpha^*})) - \log(1-\sigma(\hat Q\boldsymbol{\hat\alpha^*}))\right)\right|\\
    + \frac{1}{d}\left|\boldsymbol{Y}^\intercal \left(\log(\sigma(\hat Q\boldsymbol{\alpha^*})) - \log(\sigma(Q\boldsymbol{\alpha^*}))\right)+ (\boldsymbol{1}-\boldsymbol{Y})^\intercal \left(\log(1-\sigma(\hat Q\boldsymbol{\alpha^*})) - \log(1-\sigma( Q\boldsymbol{\alpha^*}))\right)\right|
\end{multlined}\\
&\begin{multlined}[b]
    \le\frac{1}{d}\boldsymbol{Y}^\intercal \left|\log(\sigma(Q\boldsymbol{\hat\alpha^*})) - \log(\sigma(\hat Q\boldsymbol{\hat\alpha^*}))\right|+ \frac{1}{d}(\boldsymbol{1}-\boldsymbol{Y})^\intercal \left|\log(1-\sigma(Q\boldsymbol{\hat\alpha^*})) - \log(1-\sigma(\hat Q\boldsymbol{\hat\alpha^*}))\right|\\
    + \frac{1}{d}\boldsymbol{Y}^\intercal \left|\log(\sigma(\hat Q\boldsymbol{\alpha^*})) - \log(\sigma(Q\boldsymbol{\alpha^*}))\right|+ \frac{1}{d}(\boldsymbol{1}-\boldsymbol{Y})^\intercal \left|\log(1-\sigma(\hat Q\boldsymbol{\alpha^*})) - \log(1-\sigma( Q\boldsymbol{\alpha^*}))\right|
\end{multlined}
\end{align}
where the last inequality is true due to $\boldsymbol{Y} \in \{0, 1\}^d$, and thus one can apply the triangular inequality.

We observe that the functions $\log(\sigma(\cdot))$ and $\log(1-\sigma(\cdot))$ are Lipschitz continuous, with the Lipschitz constant being
\begin{align}
    K_{\log(\sigma(\cdot))} = \sup_{x\in\mathbb{R}} \left|\frac{\partial\log(\sigma(x))}{\partial x}\right| =  \sup_{x\in\mathbb{R}} (1 - \sigma(x)) = 1,\\
     K_{\log(1-\sigma(\cdot))} = \sup_{x\in\mathbb{R}} \left|\frac{\partial\log(1-\sigma(x))}{\partial x}\right| =  \sup_{x\in\mathbb{R}} \sigma(x) = 1.
\end{align}
Hence, we see that 
\begin{align}
\Delta \mathcal{L}_{\mathrm{BCE}} \le& \frac{1}{d}\boldsymbol{Y}^\intercal \left|Q\boldsymbol{\hat\alpha^*} - \hat Q\boldsymbol{\hat\alpha^*}\right|+ \frac{1}{d}(\boldsymbol{1}-\boldsymbol{Y})^\intercal \left|Q\boldsymbol{\hat\alpha^*}- \hat Q\boldsymbol{\hat\alpha^*}\right| + \frac{1}{d}\boldsymbol{Y}^\intercal \left|\hat Q\boldsymbol{\alpha^*} - Q\boldsymbol{\alpha^*}\right|+ \frac{1}{d}(\boldsymbol{1}-\boldsymbol{Y})^\intercal \left|\hat Q\boldsymbol{\alpha^*} -  Q\boldsymbol{\alpha^*}\right|\\
=& \frac{1}{d}\boldsymbol{1}^\intercal \left|Q\boldsymbol{\hat\alpha^*} - \hat Q\boldsymbol{\hat\alpha^*}\right| + \frac{1}{d}\boldsymbol{1}^\intercal \left|\hat Q\boldsymbol{\alpha^*} - Q\boldsymbol{\alpha^*}\right| =\frac{1}{d} \left\|Q\boldsymbol{\hat\alpha^*} - \hat Q\boldsymbol{\hat\alpha^*}\right\|_1 + \frac{1}{d}\left\|\hat Q\boldsymbol{\alpha^*} - Q\boldsymbol{\alpha^*}\right\|_1\\
\le& \frac{1}{\sqrt{d}} \left(\|Q\boldsymbol{\hat\alpha^*} - \hat Q\boldsymbol{\hat\alpha^*}\|_2 + \|\hat Q\boldsymbol{\alpha^*} - Q\boldsymbol{\alpha^*}\|_2\right) \le 2\sqrt{m}\|\hat{Q}-Q\|_{\max}
\end{align}
which returns the same element-wise bound $\|\hat{Q}-Q\|_{\max} < \frac{\epsilon}{2\sqrt{m}}$.

\twocolumngrid
\bibliography{main}

\end{document}